\documentclass[submission,copyright,creativecommons]{eptcs}
\providecommand{\event}{LSFA 2025} 

\usepackage{iftex}

\ifpdf
  \usepackage{underscore}         
  \usepackage[T1]{fontenc}        
\else
  \usepackage{breakurl}           
\fi

\usepackage[dvipsnames]{xcolor}
\usepackage{comment}
\usepackage{graphicx}
\usepackage{amsthm} 
\usepackage{tikz}
\usetikzlibrary{arrows,decorations,shapes,automata,positioning,decorations.pathmorphing}
\tikzset{
	every picture/.style = {
		thick,
		>=stealth',
		node distance = 1.5em and 3em,
	}
	,
	cross line/.style = {
		preaction = {
			draw=white,
			-,
			line width=4pt
		}
	}
	,
	state/.style = {
		rectangle,
		rounded corners = 5pt,
		font = \footnotesize,
		draw,
		minimum width = 1em,
		minimum height = 1em
	}
	, 
	label-state/.style = {
		sloped,
		font = \scriptsize,
		label distance = -2pt
	}
	, 
	label-edge/.style = {
		font = \scriptsize,    
		label distance = -2pt
	}
}

\usepackage{amsmath}
\usepackage{bold-extra}
\usepackage{amsmath}
\usepackage{amssymb}
\usepackage{prftree}
\usepackage{xspace}
\usepackage{
  url,
  colortbl,
  stmaryrd,
  hyperref,
  cleveref,
  xargs,
  booktabs,
  multicol,
  mathtools}

\newtheorem{theorem}{Theorem}
\newtheorem{proposition}[theorem]{Proposition}
\newtheorem{lemma}[theorem]{Lemma}

\newtheorem{definition}[theorem]{Definition}

\theoremstyle{remark}
\newtheorem{example}{Example}
 
\usepackage{arydshln}
\usepackage{etoolbox}
\newcommand{\zerodisplayskips}{%
  \setlength{\abovedisplayskip}{5pt}%
  \setlength{\belowdisplayskip}{5pt}%
  \setlength{\abovedisplayshortskip}{5pt}%
  \setlength{\belowdisplayshortskip}{5pt}}
\appto{\normalsize}{\zerodisplayskips}
\appto{\small}{\zerodisplayskips}
\appto{\footnotesize}{\zerodisplayskips}

\Crefname{definition}{Def.}{Defs.}
\Crefname{figure}{Fig.}{Figs.}
\Crefname{theorem}{Thm.}{Thms.}
\Crefname{section}{Sec.}{Secs.}
\Crefname{appendix}{Appx.}{Apps.}
\Crefname{remark}{Remark}{Remarks}

\usepackage[colorinlistoftodos,prependcaption,textsize=tiny]{todonotes}

\newcommand{\midd}{\mathrel{\,\mid\,}}

\newcommand{\<}{\langle}
\renewcommand{\>}{\rangle}

\newcommand{\tup}[1]{\langle #1 \rangle}

\DeclarePairedDelimiterX\setof[2]{\{}{\}}{\,#1 \;\delimsize\vert\; #2\,}

\newcommand{\hxpd}{\text{\rm HXPath}_{\rm D}}

\newcommand{\hylo}{\mathscr{H}\xspace}

\newcommand{\rr}[1]{{\color{WildStrawberry}#1}}
\newcommand{\cc}[1]{{\color{ProcessBlue}#1}}
\newcommand{\pp}[1]{{\color{Mulberry}#1}}

\newcommand{\taag}[3]{@_{#2}\tup{#1}{#3}}

\usepackage{ifthen}
\newcommand{\cmprif}[1]{
    \ifthenelse{\equal{#1}{=}}
    {
      =_{\compc}
    }
    {
      \ifthenelse{\equal{#1}{\neq}}
      {
        \neq_{\compc}
      }
      {
        \mathrel{#1}
      }
    }
}

\newcommand{\tagg}[3]{\tup{{#2}{:} \cmprif{#1} {#3}{:}}}

\newcommand{\hilbert}{\mathbf{H}}
\newcommand{\gentzen}{\mathbf{G}}
\newcommand{\rgentzen}{\gentzen'}
\newcommand{\tgentzen}{\gentzen_{\hylo(@)}}

\newcommand{\liff}{\leftrightarrow}

\newcommand{\prop}{\mathsf{Prop}}

\newcommand{\nom}{\mathsf{Nom}}
\newcommand{\mo}{\mathsf{Mod}}

\newcommand{\cmp}{\mathsf{Cmp}}

\newcommandx{\dow}[1][1={a}]{\xspace\mathsf{#1}\xspace}
\newcommand{\dowa}{\xspace\mathsf{a}\xspace}

\newcommand{\compc}{\xspace\mathsf{c}\xspace}

\newcommand{\cmpr}{\mathrel{\blacktriangle}}
\newcommand{\cmpd}{\mathrel{\blacktriangledown}}

\newcommandx{\amodel}[1][1={M}]{\mathfrak{#1}}

\newcommand{\nodes}{\mathrm{N}}
\newcommand{\node}{n}
\newcommand{\R}{\mathrm{R}}
\newcommand{\deq}{\approx}
\newcommand{\dneq}{\not\approx}
\newcommand{\V}{\mathrm{V}}

\renewcommand{\iff}{\mathit{iff}}

\newcommand{\Rho}{\mathrm{P}}

\newcommand{\inv}[1]{#1^{\text{-}1}}

\usepackage{tikz}
\usetikzlibrary{arrows,calc,patterns,positioning,shapes}
\usetikzlibrary{decorations.pathmorphing}
\tikzset{
  modal/.style={>=stealth',shorten >=1pt,shorten <=1pt,auto},
  world/.style={circle,draw,minimum size=1.3cm},
  point/.style={circle,draw,fill=black,inner sep=0.2mm},
  reflexive/.style={->,in=120,out=60,loop,looseness=#1},
  reflexive/.default={5},
  reflexive point/.style={->,in=135,out=45,loop,looseness=#1},
  reflexive point/.default={25},
  label-edge/.style = {font = \footnotesize},
}
\tikzset{
  reflexive above/.style={->,loop,in=120,out=60,looseness=#1},
  reflexive above/.default={7},
  reflexive below/.style={->,loop,in=240,out=300,looseness=#1},
  reflexive below/.default={7},
  reflexive left/.style={->,loop,in=150,out=210,looseness=#1},
  reflexive left/.default={7},
  reflexive right/.style={->,loop,in=30,out=330,looseness=#1},
  reflexive right/.default={7}
}

\usepackage{mathrsfs}

\newcommandx{\aderivation}[1][1={D}]{\mathscr{#1}}

\DeclareMathOperator{\size}{size}

\newcommand{\arule}{\mathrm{P}}

\newcommand{\thescalefactor}{.83}

\title{Sequent Calculi for Data-Aware Modal Logics}

\author{Carlos Areces
	\institute{Universidad
	Nacional de C{\'o}rdoba} 
	\institute{and CONICET, Argentina}
	\email{carlos.areces@unc.edu.ar}
\and
Valentin Cassano
	\institute{Universidad
	Nacional de R\'io Cuarto}
    \institute{and CONICET, Argentina}
	\email{valentin@dc.exa.unrc.edu.ar}
\and	
Danae Dutto
	\institute{Universidad Nacional de
	C{\'o}rdoba}
	\institute{and CONICET, Argentina} 
	\email{ddutto@dc.exa.unrc.edu.ar}
\and
Raul Fervari
	\institute{Universidad
	Nacional de C{\'o}rdoba} 
	\institute{and CONICET, Argentina} 
	\email{rfervari@unc.edu.ar}
}

\def\titlerunning{Sequent Calculi for Data-Aware Modal Logics}
\def\authorrunning{C.\ Areces, V.\ Cassano, D.\ Dutto and R.\ Fervari}
\begin{document}
\maketitle

\begin{abstract}
	Data-aware modal logics offer a powerful formalism for reasoning about semi-structured queries in languages such as DataGL, XPath, and GQL.
	In brief, these logics can be viewed as modal systems capable of expressing both reachability statements and data-aware properties, such as value comparisons.
	One particularly expressive logic in this landscape is $\hxpd$, a hybrid modal logic that captures not only the navigational core of XPath but also data comparisons, node labels (keys), and key-based navigation operators.
	While previous work on $\hxpd$ has primarily focused on its model-theoretic properties, in this paper we approach $\hxpd$ from a proof-theoretic perspective.
	Concretely, we present a sound and complete Gentzen-style sequent calculus for $\hxpd$. Moreover, we show all rules in this calculus are invertible, and that it enjoys cut elimination.
	Our work contributes to the proof-theoretic foundations of data-aware modal logics, and enables a deeper logical analysis of query languages over graph-structured data.
	Moreover, our results lay the groundwork for extending proof-theoretic techniques to a broader class of modal systems.
\end{abstract}

\section{Introduction}
\label{sec:intro}

Semi-structured data organization~\cite{Buneman97} is reemerging as a flexible paradigm for representing and storing information, offering a contrast to the rigid rows-and-columns model of traditional relational databases.
This approach is exemplified by \emph{graph databases}, where information is represented and stored in a semi-structured form as \emph{data graphs}, in which nodes and edges are  labeled with values that capture both structure and data content~\cite{Robinson:2013}.
Graph databases are widely adopted in domains such as the web, social networks, fraud detection, and recommendation engines, and are implemented in tools like \href{https://neo4j.com/}{Neo4j} and \href{http://aws.amazon.com/neptune/}{Amazon Neptune}. 
To support such range of applications, graph databases rely on expressive query languages designed to navigate and extract information from their underlying structures.
However, unlike relational databases, where SQL serves as the \emph{de facto} standard, query languages for graph databases are still evolving.
Recent efforts, as discussed in~\cite{FrancisGGLMMMPR23,FrancisGGLMMMPR23b}, are converging toward the Graph Query Language (GQL) Standard~\cite{ISO39075}, which aims to unify existing industrial approaches.
GQL builds on its academic predecessor G-CORE~\cite{AnglesABBFGLPPS18} and draws conceptual inspiration from Regular Path Queries~\cite{LV12} and the XML Path Language (XPath)~\cite{LibkinMV16}, offering expressive capabilities for navigating graph structures and performing equality and inequality tests on data.
This connection highlights the importance of logical foundations in understanding and advancing modern query languages for graph data.

In this article, we focus on XPath, approaching it from a logical perspective and building on prior work that formalizes fragments of this language as modal logics.
Below, we outline a brief timeline highlighting key developments in this direction.

Initial foundational results~\cite{BK08,CateM09}, showed that the navigational fragment of XPath---known as Core-XPath~\cite{GKP05}---is strictly less expressive than Propositional Dynamic Logic (PDL) over trees.
This set the stage for subsequent work, ending with a complete axiomatization for Core-XPath in~\cite{CateM09,cateLM10}. 
As the study of XPath progressed, attention turned to features that go beyond pure navigation.
Querying data graphs often requires the ability to reason about and compare data values---a limitation of purely navigational logics.
How to properly represent and manipulate data within a modal setting gave rise to the development of \emph{data-aware modal logics}: extensions of modal logic that incorporate operators for handling data alongside navigation.
These logics are tailored to capturing the interplay between structure and data in graph-based systems, integrating mechanisms such as data comparison, binding, and quantification within the modal framework.
This includes the development of expressive yet decidable languages, the design of automata-theoretic characterizations, and the identification of suitable semantics that balance expressive power with computational tractability (see e.g.,~\cite{FigPhD}).
The result is a family of expressive formalisms capable of modeling sophisticated queries over data-rich systems.

These efforts have opened up new avenues in both theoretical studies and practical applications, particularly in areas such as verification, knowledge representation, and query languages for graph databases.
For instance,~\cite{BojanczykMSS09} introduces Core-Data-XPath as an extension of Core-XPath with equality and inequality comparisons.
The satisfiability problem and complexity aspects of fragments of these logics are explored in~\cite{FigueiraS11,Figueira12ACM,LibkinMV16,FigueiraD:2018}, while the model theory and bisimulation algorithms for this kind of logics are studied in~\cite{ADF14,ICDT14,ICDT14Jair,KR16,ADF17IC,AbriolaBFF18}.
Axiomatic systems extending those from \cite{cateLM10} are presented in~\cite{AbriolaDFF17,AbriolaFG24}.
This theoretical groundwork led to concrete applications. For example, \cite{GKP05} discusses newly proposed algorithms with existing XPath processors, while the satisfiability methods in~\cite{FigueiraD:2018} can be used to check consistency in XML documents. Furthermore, the bisimulation algorithms in~\cite{ICDT14Jair,KR16} have applications in database minimization, while~\cite{BaeldeLS19} examines how the complexity of fragments impacts real-world applications.

The work in~\cite{ArecesF21} introduces a novel extension to data-aware modal logics by incorporating \emph{keys}---or labels---for nodes, and \emph{jump-to-key} operations.
This results in the logic we refer to as $\hxpd$, which extends Core-Data-XPath with expressive constructs from Hybrid Logic. In particular, nominals and the satisfiability operator $@$, well-studied in Hybrid Logics~\cite{arec:hybr05b}, provide the means to refer directly to specific nodes, and to ``navigate'' to specific nodes.
This addition enables $\hxpd$ to express properties and navigation patterns not previously captured in earlier formalisms.
Furthermore, the hybrid features of $\hxpd$ support a novel axiomatization, with completeness established using Henkin-style models.

In this article, we contribute to the logical study of data-aware modal logics from a \emph{proof-theoretic} perspective by developing a sequent calculus for $\hxpd$.
We establish the completeness of our calculus by leveraging the completeness results in~\cite{ArecesF21}.
We also show that our calculus enjoys cut-elimination and invertibility of rules.
Furthermore, we relate our work to research on labeled modal sequent calculi with equality~\cite{Negri14} and the proof theory of hybrid logic~\cite{Torben2011}.
Developing a proof theory for modal logics is notoriously difficult---see, e.g., the discussion in~\cite{Negri2011-NEGPTF}.
Several techniques, such as display calculi~\cite{Wansing2002}, hypersequents~\cite{Avron1996}, deep inference~\cite{StewartS04}, labeled systems~\cite{Negri2005}, and hybridization~\cite{Torben2011}, have been explored to establish a solid proof-theoretic foundation for various modal logics.
Our results suggest that labeling and hybridization can also support a well-grounded proof theory for data-aware modal logics.


\smallskip
\noindent \textbf{Outline.}
\Cref{sec:xpath} introduces the logic $\hxpd$.
\Cref{sec:calculus} presents the sequent calculus $\gentzen$ for $\hxpd$.
Key properties of $\gentzen$---soundness, completeness, and invertibility of rules---are addressed in~\Cref{sec:soundness,sec:completeness,sec:invertibility}.
Cut elimination is established in~\Cref{sec:cut}.
\Cref{sec:torben} demonstrates a correspondence between $\gentzen$ and the sequent calculus for Basic Hybrid Logic from~\cite{Torben2011}.
\Cref{sec:final} concludes with final remarks and directions for future work.
Details and omitted proofs are available in~\cite{arxiv}.

\section{Hybrid XPath with Data}
\label{sec:xpath}

We begin by presenting the formal syntax and semantics of $\hxpd$---further details are found in~\cite{ArecesF21}.
In our presentation, we assume $\prop$, $\nom$, $\mo$, and $\cmp$ be pairwise disjoint sets of symbols for propositions, nominals, modalities, and data comparisons, respectively.
Moreover, we assume that $\mo$ and $\cmp$ are finite, while $\prop$ and $\nom$ are countably infinite.

\begin{definition}\label[definition]{def:lang}
    The language of $\hxpd$ has \emph{path} expressions (denoted $\alpha$, $\beta$, \dots) and \emph{node} expressions (denoted $\varphi$, $\psi$, \dots), mutually defined by the grammar:
    \begin{align*}
        \alpha, \beta \coloneqq\ &
        \dowa \midd
        i{:} \midd
        \varphi? \midd
        \alpha\beta
        \\
        \varphi, \psi \coloneqq\ &
            p \midd
            i \midd
            \bot \midd
            \varphi \to \psi \midd
            @_i\varphi \midd
            \tup{\dowa}\varphi \midd
            \<\alpha =_{\compc} \beta\> \midd \<\alpha \neq_{\compc} \beta\>,
    \end{align*}
    where 
    $p \in \prop$,
    $i \in \nom$,
    $\dowa \in \mo$, 
    and 
    $\compc\in\cmp$.\footnote{%
        Unlike~\cite{ArecesF21}, we treat $@_i\varphi$ and $\langle\dowa\rangle\varphi$ as primitive. This choice simplifies the formulation of the sequent calculus in~\Cref{sec:calculus}.
    }
%
%
    For path expressions, we use $\epsilon \coloneqq \top?$ to indicate the \emph{empty}~path.
    For node expressions, we use standard abbreviations:
        $\top \coloneqq \bot \to \bot$,
        $\lnot \varphi \coloneqq \varphi \to \bot$,
        $\varphi \lor \psi \coloneqq \lnot\varphi \to \psi$,
        $\varphi \land \psi \coloneqq \lnot(\varphi \to \lnot\psi)$, and
        $\varphi \liff \psi \coloneqq (\varphi \to \psi) \land (\psi \to \varphi)$.
    We also abbreviate:
        $\tup{j{:}}\varphi \coloneqq @_j\varphi$,
        $\tup{\psi?}\varphi \coloneqq \psi \land \varphi$,
        $\tup{\alpha\beta}\varphi \coloneqq \tup{\alpha}\tup{\beta}\varphi$, and
        $[\alpha]\varphi \coloneqq \lnot\tup{\alpha}\lnot\varphi$.
    These abbreviations complete all cases to allow arbitrary paths $\alpha$ inside a unary modality $\tup{\alpha}$. 
    Finally, we abbreviate
        $[\alpha \cmpr \beta] \coloneqq \lnot\tup{\alpha \cmpd \beta}$.
    In the last abbreviation, we use $\cmpr$ when there is no need to distinguish $=_{\compc}$ and $\neq_{\compc}$, and use ${\cmpd}$ to indicate ${\neq_{\compc}}$ if ${\cmpr}$ is ${=_{\compc}}$, and to indicate ${=_{\compc}}$ if ${\cmpr}$ is ${\neq_{\compc}}$.
\end{definition}

Path and node expressions are interpreted over hybrid data models.

\begin{definition}\label[definition]{def:models}
    A \emph{(hybrid data) model} is a tuple 
$\amodel = \tup{  
        \nodes,
        \{\R_{\dowa}\}_{\dowa \in \mo},
        \{\deq_{\compc}\}_{\compc \in \cmp}, 
        g,
        \V}$,
    where
        $\nodes$ is a non-empty set of nodes;
        each $\R_{\dowa}$ is a (binary) \emph{accessibility relation} on $\nodes$;
        each $\deq_{\compc}$ is an equivalence relation on $\nodes$, called a \emph{comparison};
        $g : \nom \to \nodes$ is a \emph{nominal assignment}; and 
        $\V : \prop \to 2^{\nodes}$ is a \emph{valuation}. 
\end{definition}

The satisfiability relation for path and node expressions is as follows.

\begin{definition}\label[definition]{def:semantics} 
    Let $\amodel=\tup{\nodes, \{\R_{\dowa}\}_{\dowa \in \mo},\{\deq_{\compc}\}_{\compc \in \cmp}, g,\V}$ be a model, and let ${\{\node, \node'\} \subseteq \nodes}$. 
    The satisfiability relation $\Vdash$ is given by the following conditions: 
        \[
        \begin{array}{l @{~~}c@{~~} l}
            \amodel,\node,\node' \Vdash \dowa
            & \iff & \node \R_{\dowa} \node' \\
    
            \amodel,\node,\node' \Vdash i{:}
            & \iff &  g(i) = \node' \\
    
            \amodel,\node,\node' \Vdash \varphi?
            & \iff &  \node = \node' \mbox{ and } \amodel,\node \Vdash \varphi \\
    
            \amodel,\node,\node' \Vdash \alpha\beta
            & \iff & \mbox{exists} ~ \node'' \in \nodes \mbox{ s.t. }
                        \amodel,\node,\node'' \Vdash \alpha \mbox{ and } \amodel,\node'',\node'\Vdash \beta \\
                          	
            \amodel,\node \Vdash p
            & \iff & \node \in \V(p) \\
            
            \amodel,\node \Vdash i
            & \iff & g(i) = \node  \\

            \amodel,\node \Vdash \bot
            &\multicolumn{2}{@{}l}{\text{never}} \\
            
            \amodel,\node \Vdash \varphi\to\psi
            & \iff & \amodel,\node \Vdash \varphi \mbox{ implies } \amodel,\node \Vdash \psi \\

            \amodel,\node \Vdash @_i \varphi
            & \iff & \amodel, g(i)\Vdash \varphi\\


            \amodel,\node \Vdash \tup{\dowa}\varphi 
            & \iff & \mbox{exists $\node'\in\nodes$ s.t.\ } \amodel,\node,\node'\Vdash \dowa \mbox{ and } \amodel,\node'\Vdash \varphi \\ 

            \amodel,\node \Vdash \<\alpha=_{\compc}\beta\>
            & \iff & \mbox{exists $\node',\node''\in \nodes$ s.t.\ } 
                        \amodel,\node,\node' \Vdash \alpha, \ 
                        \amodel,\node,\node'' \Vdash \beta \mbox{ and }
                        \node' \deq_{\compc} \node''
        \\
            \amodel,\node\Vdash\<\alpha\neq_{\compc}\beta\>
            & \iff & \mbox{exists $\node',\node''\in \nodes$ s.t.\ } 
                        \amodel,\node,\node' \Vdash \alpha, \ 
                        \amodel,\node,\node'' \Vdash \beta \mbox{ and } \node' \dneq_{\compc} \node''.\\
        \end{array}
        \]
    Let $\Psi$ be a set of node expressions, we use $\amodel,\node\Vdash\Psi$ to indicate $\amodel,\node\Vdash\psi$ for all $\psi\in\Psi$. 
    We say $\Psi$ is \emph{satisfiable} iff there exists $\amodel,\node$ s.t.\ $\amodel,\node \Vdash \Psi$.
    We call a node expression $\varphi$ a \emph{consequence} of $\Psi$, written $\Psi \vDash \varphi$, iff $\Psi\cup\{\lnot \varphi\}$ is unsatisfiable.
    If $\Psi$ is the empty set, we write $\vDash \varphi$ and call $\varphi$ a \emph{tautology}.
\end{definition}

\Cref{prop:abbrv}, which follows directly from \Cref{def:semantics}, confirms that the abbreviations behave as intended.

\begin{proposition}\label[proposition]{prop:abbrv}
$\amodel,\node \Vdash \<\alpha\>\varphi$ iff $\mbox{exists $\node' \in \nodes$ s.t.},~
            \amodel,\node,\node' \Vdash\alpha$ and  
            $\amodel,\node' \Vdash \varphi$.
    Moreover, $\amodel,\node \Vdash [\alpha=_{\compc}\beta]$ iff for all $\node',\node''\in \nodes$,  
                $\amodel,\node,\node' \Vdash \alpha$ and 
                $\amodel,\node,\node'' \Vdash \beta$ imply  
                $\node' \deq_{\compc} \node''$.
    Finally,
    $\amodel,\node\Vdash[\alpha\neq_{\compc}\beta]$ iff for all $\node',\node''\in \nodes$, 
                $\amodel,\node,\node' \Vdash \alpha$ and 
                $\amodel,\node,\node'' \Vdash \beta$ imply  
                $\node' \dneq_{\compc} \node''$.
\end{proposition}

Some intuitive remarks are helpful here. Path expressions indicate reachability relations on nodes, i.e., nodes $n$ and $n'$ satisfy the path expression $\alpha$ if the pair $(n,n')$ is in the relation defined by $\alpha$; or, equivalently, if $n'$ is reachable from $n$ following an $\alpha$ path. The test expression $\varphi?$ checks whether the node expression $\varphi$ holds in a given point in a path. The jump-to-key expression $i{:}$ resets the current path to start from the node whose key is $i$. Node expressions $p$, $i$, $\bot$, $\varphi \to \psi$, $@_i\varphi$, and $\tup{\dowa}\varphi$, have classical readings (see~\cite{arec:hybr05b}). The novel data comparison operator $\<\alpha=_{\compc}\beta\>$ (resp. $\<\alpha \not=_{\compc}\beta\>$) requires the existence of paths $\alpha$ and $\beta$, from the current point of evaluation, whose end nodes are in the relation $\deq_{\compc}$ (resp.~$\dneq_{\compc}$) relation. This mirrors how XPath performs data comparisons over data graphs.

We conclude the presentation of $\hxpd$ by illustrating how the logic functions as a semi-structured query language over data graphs.
The idea is that hybrid data models represent data graphs, while node expressions represent queries that specify both structural properties (e.g., reachability) and data-aware conditions (e.g., comparisons between node values).
The following example clarifies this intuition.

\begin{figure}[t]
	\begin{center}
		\begin{tikzpicture}[->, every node/.style={scale=0.95}]
			
			\node [state, label = {[label-state]above left:$i_1$}] (w1) {$\begin{array}{c}Person \\  
					{\sf name}: \texttt{"Alice"}\end{array}$};      
			\node [state, right = of w1] (w2) {$\begin{array}{c}Person \\  
					{\sf name}: \texttt{"Bob"}\end{array}$};  
			
			\node [state,  right = of w2, label = {[label-state]above left:$i_2$}] (w3) {$\begin{array}{c}Person \\  
					{\sf name}: \texttt{"Alice"}\end{array}$};      

			\node [state, below = of w1] (ww1) {$\begin{array}{c}Date\\  
					{\sf val}: \texttt{1977-07-07}\end{array}$};      
			\node [state, below = of w2] (ww2) {$\begin{array}{c}Date\\  
					{\sf val}: \texttt{1977-07-07}\end{array}$};      
			
			\node [state,  below = of w3] (ww3) {$\begin{array}{c}Date\\  
					{\sf val}: \texttt{1975-05-05}\end{array}$};      
			
			
			\path (w1) edge[<->] node [label-edge, above] {$\textbf{friends}$} (w2);
			\path (w2) edge[<->] node [label-edge, above] {$\textbf{friends}$} (w3);   
			
			\path (w1) edge node [label-edge, right] {$\textbf{born}$} (ww1);
			\path (w2) edge node [label-edge, right] {$\textbf{born}$} (ww2);
			
			\path (w3) edge node [label-edge, right] {$\textbf{born}$} (ww3);
			
		\end{tikzpicture}
	\end{center}
	\caption{Example of a Data Graph}\label{fig:datagraph}
\end{figure}

\begin{example}\label{example1}
    \Cref{fig:datagraph} shows a simple example of a data graph.  We will not formally define data graphs here, see~\cite{ArecesF21} for further details.  Intuitively, we can see that the graph encodes information about people and their birth dates, and it also includes a friendship relation. Notice that the graph encodes information of different types: propositional information (certain nodes are defined as {Persons}, others are {Dates}), relational information (e.g., the {\rm \textbf{born}} relation encodes the birthdate of a person), and concrete data encoded as pairs \emph{attribute:data} (e.g., one of the persons in the data graph is named {\rm \texttt{"Alice"}}).  Moreover, the labels $i_1$ and $i_2$ are indexes, that allows direct access to certain nodes in the data graph. 

    The information in the data graph shown in \Cref{fig:datagraph} can be directly encoded as a hybrid data model. For the structure of the data graph this is immediate. The propositional content in the data graph is captured using proposition symbols, while node indices are managed via the nominal assignment. Crucially, concrete data values are abstracted and represented through data comparisons. Namely, for each attribute~\texttt{c}, we introduce in the model a comparison relation $\deq_{\compc}$ such that $n_1 \deq_{\compc} n_2$ iff there is a data value \texttt{d} such that \textsf{c} : \texttt{d} (as shown in the figure) is both in $n_1$ and $n_2$---that is, $n_1$ and $n_2$ have the same value \texttt{d} in attribute \textsf{c}.
    The following hybrid data model corresponds to the data graph in \Cref{fig:datagraph}.
    \begin{align*}
    {\rm N} & = \{n_1, n_2, n_3, n_4, n_5, n_6\}
        & {\rm V}(Person) & = \{n_1,n_2,n_3\}
            & \deq_{\sf name} & = \{(n_1,n_3)\}^{\sf eq}
    \\
    {\rm R}_{\rm \bf friends} & = \{(n_1,n_2),(n_2,n_1),(n_2,n_3),(n_3,n_2)\}
        & {\rm V}(Date) & = \{n_4,n_5,n_6\}
            & \deq_{\sf val} & = \{(n_4,n_5)\}^{\sf eq}
    \\
    {\rm R}_{\rm \bf born} & = \{(n_1,n_4),(n_2,n_5),(n_3,n_6)\}
        & g(i_1) & =  n_1     & g(i_2) & = n_3
    \end{align*}
    We use $R^{\sf eq}$ to denote the closure of $R$ under reflexivity, symmetry and transitivity.
    The equivalence relations $\deq_{\sf name}$ and $\deq_{\sf val}$ are sufficient for the type of data operations permitted in $\hxpd$.

    With the hybrid data model in place, we now turn to how node expressions act as queries.
    The table below presents a few example properties, along with their formalizations as node expressions:
    \begin{table}[h]
        \footnotesize
        \centering
        \begin{tabular}{p{5cm}p{10cm}}
        \toprule
            \textsc{Property} & \textsc{Formalization} \\
            \midrule
            The person with key $i_1$ is a friend of someone born on the same day. 
            & $\tup{i_1{:}\,\textbf{born}\,(\mathit{Date}?) =_{\sf val} i_1{:}\,\textbf{friends}\,\textbf{born}\,(\mathit{Date}?)}$ \\
            \midrule
            The person with key $i_2$ shares a birth date with none of her friends. 
            & $[i_2{:}\,\textbf{born}\,(\mathit{Date}?) \neq_{\sf val} i_2{:}\,\textbf{friends}\,\textbf{born}\,(\mathit{Date}?)]$ \\
            \midrule
            The persons with keys $i_1$ and $i_2$ have the same name but different birth dates. 
            & ${\tup{i_1{:}\,(\mathit{Person}?) =_{\sf name} i_2{:}\,(\mathit{Person}?)}}\land \tup{i_1{:}\,\textbf{born}\,(\mathit{Date}?) \neq_{\sf val} i_2{:}\,\textbf{born}\,(\mathit{Date}?)}$ \\
            \bottomrule
        \end{tabular}
        \end{table}

    \noindent
    Each of the node expressions above evaluates to true at every point in the hybrid data model constructed from the example graph database. This showcases how $\hxpd$ combines structural navigation with data-aware comparisons to express rich, semantically meaningful queries.
\end{example}

To sum up, $\hxpd$ is of interest both in practice and in theory. On the practical side, it faithfully captures a fragment of XPath involving data-aware queries and key-based navigation. On the theoretical side, it poses new challenges for modal logic: its comparison modalities are complex as they combine navigation modalities with (in)equality reasoning.
Despite these challenges, in the next section we show that an elegant sequent calculus for $\hxpd$ can indeed be defined.

\section{A Sequent Calculus for $\hxpd$}
\label{sec:calculus}

In this section, we present our Gentzen-style sequent calculus $\gentzen$ for $\hxpd$.
The calculus draws inspiration from the modal system with equality developed in~\cite{Negri14}, but extends it to accommodate the distinctive features of $\hxpd$: nominals, satisfiability modalities, path expressions, and data comparisons.
In brief, we build $\gentzen$ on sequents $\Gamma \vdash \Delta$, where $\Gamma, \Delta$ is a set of node expressions in $\hxpd$.
Sequents capture the idea that if all node expressions in $\Gamma$ hold, then at least one node expression in $\Delta$ must also hold.
The inference rules of $\gentzen$ systematically decompose the expressions in sequents, closely mirroring the semantic behavior of each logical connective.
Out aim is to construct a proof system that addresses the entailment problem for the logic, and to provide a rule-based decomposition of the meaning of its logical connectives.
As we will see also, the calculus is particularly well-suited for structural analysis of proofs.
With the basic ideas in place, let us begin by: defining sequents, introducing the inference rules, and outlining the construction of derivations in~$\gentzen$.

\begin{definition}\label[definition]{def:sequent}
    A \emph{sequent} is a pair $\Gamma \vdash \Delta$, where $\Gamma$ and $\Delta$ are finite, possibly empty sets of node expressions of the form $\tagg{\cmpr}{i}{j}$ or $@_i\varphi$.\footnote{Following standard notation, we will not use curly brackets when writing down sequents.}
    In a sequent $\Gamma \vdash \Delta$, the set $\Gamma$ is called the \emph{antecedent} and the set $\Delta$ is called the \emph{consequent}.
    In turn, a \emph{rule} is a pair $(\Gamma \vdash \Delta, \{\Gamma_1 \vdash \Delta_1, \dots, \Gamma_n \vdash \Delta_n\})$, where $\Gamma \vdash \Delta$ is called the \emph{conclusion}, and $\{\Gamma_1 \vdash \Delta_1, \dots, \Gamma_n \vdash \Delta_n\}$ the set of \emph{premisses} of the rule.  If the set of premisses is empty, we say that the rule is an \emph{axiom}.
    In each rule, certain node expressions in the conclusion sequent are designated as \emph{principal}.
    These are the node expressions the rule acts upon. The remaining node expressions in the sequent are the \emph{context} of the rule, and are carried unchanged across its application. This distinction enables the rule to isolate and analyze the logical structure of the principal formulas while treating the context uniformly.
    The rules of $\gentzen$ are shown in \Cref{rules:hxpd}, in a standard format---i.e., the conclusion is shown under a line, the premisses are above the line in no particular order, the principal node expressions and the context are identified immediately from the presentation.
\end{definition}

With sequents and rules defined, we proceed to formalize the notion of a derivation in $\gentzen$.

\begin{definition}\label[definition]{def:derivation}
    A \emph{derivation} in $\gentzen$ is a sequent-labeled tree constructed according to the rules in \Cref{rules:hxpd}. More precisely,  each non-leaf node in the tree is a sequent that is the conclusion of a rule in $\gentzen$, its immediate successors in the tree (going upwards) are the premisses of that rule. 
    The root of the tree is called the \emph{end-sequent} of the derivation.
    A sequent is \emph{derivable} if it is the end-sequent of some derivation, and it is \emph{provable} if it has a derivation whose leaves are all instances of (Ax)~or~ ($\bot$).
    We use $\Gamma \vdash_{\gentzen} \Delta$ to indicate that $\Gamma \vdash \Delta$ is provable.
A rule is \emph{derived} iff there is a derivation of the conclusion of the rule whose leaves are either axioms or belong to the premisses of the rule.

%
%
    %
\end{definition}

\begin{figure}[h!]
   \centerline
  {
  \scalebox{\thescalefactor}
  {
  \begin{tabular}{@{}c@{}}
  
  \toprule
  \textbf{\textsc{Propositional Rules}} \\
  \midrule

  \tabularnewline[-7pt]
  \begin{tabular}{@{}c@{}}
    \prfbyaxiom{\footnotesize(Ax)$^{\ddagger}$}{\varphi, \Gamma \vdash \Delta, \varphi} 
    
    ~

    \prfbyaxiom{\footnotesize($\bot$)}{@_i\bot, \Gamma \vdash \Delta}
    
    ~

    \prftree[r]{\footnotesize($\to$L)}
    {\Gamma \vdash \Delta, @_i \varphi}{@_i \psi, \Gamma \vdash \Delta}
    {@_i(\varphi \to \psi), \Gamma \vdash \Delta}
    
    ~
    
    \prftree[r]{\footnotesize($\to$R)}
    {@_i \varphi, \Gamma \vdash \Delta, @_i\psi}
    {\Gamma \vdash \Delta, @_i(\varphi \to \psi)}

    \tabularnewline[5pt]
    {\footnotesize$^{\ddagger}$ $\varphi$ is of the form $@_ip$, $@_ij$, or $\tup{i{:} =_{\compc} j{:}}$}
  \end{tabular}

  \tabularnewline

  \midrule
  \textbf{\textsc{Rules for Nominals}} \\
  \midrule

  \tabularnewline[-7pt]
  \begin{tabular}{@{}c@{}}

    \prftree[r]{\footnotesize($@$T)}
    {@_ii, \Gamma \vdash \Delta}
    {\Gamma \vdash \Delta}
    
    \quad

    \prftree[r]{\footnotesize($@5$)}
    {@_jk, @_ij, @_ik, \Gamma \vdash \Delta}
    {@_ij, @_ik, \Gamma \vdash \Delta}

    \quad

    \prftree[r]{\footnotesize(Nom)$^{\dagger}$}
    {@_ij, \Gamma \vdash \Delta}
    {\Gamma \vdash \Delta}
    
    \tabularnewline[7pt]

    \prftree[r]{\footnotesize(S$_1$)$^{\ddagger}$}
    {@_j\varphi, @_ij, @_i\varphi, \Gamma \vdash \Delta}
    {@_ij, @_i\varphi, \Gamma \vdash \Delta}    

    \quad

    \prftree[r]{\footnotesize(S$_2$)}
    {\taag{\dowa}{i}{k}, @_jk, \taag{\dowa}{i}{j}, \Gamma \vdash \Delta}
    {@_jk, \taag{\dowa}{i}{j}, \Gamma \vdash \Delta}

    \quad

    \prftree[r]{\footnotesize(S$_3$)}
    {\tagg{=}{j}{k}, @_ij, \tagg{=}{i}{k}, \Gamma \vdash \Delta}
    {@_ij, \tagg{=}{i}{k}, \Gamma \vdash \Delta}

    \tabularnewline[5pt]
    {\footnotesize $^{\dagger}$ $j$ is not in the conclusion}
    \quad
    {\footnotesize $^{\ddagger}$ $\varphi$ is of the form $p$, $\bot$, or $\tup{\dowa}k$}
  \end{tabular}

  \tabularnewline

  \midrule
  \textbf{\textsc{Rules for Modalities}} \\
  \midrule

  \tabularnewline[-7pt]
  \begin{tabular}{@{}c@{}}
    \prftree[r]{\footnotesize($@$L)}
    {@_i\varphi, \Gamma \vdash \Delta }
    {@_j@_i\varphi, \Gamma \vdash \Delta}

    \quad

    \prftree[r]{\footnotesize($@$R)}
    {\Gamma \vdash \Delta,@_i\varphi}
    {\Gamma \vdash \Delta,@_j@_i\varphi}
    
    \quad
    
    \prftree[r]{\footnotesize(${\tup{\dowa}}$L)$^{\dagger}$}
    {@_i\tup{\dowa}j, @_j\varphi, \Gamma \vdash \Delta}
    {@_i\tup{\dowa}\varphi, \Gamma \vdash \Delta}
    
    \quad
    
    \prftree[r]{\footnotesize(${\tup{\dowa}}$R)}
    {@_i\tup{\dowa}j, \Gamma \vdash \Delta, @_i\tup{\dowa}\varphi, @_j\varphi}
    {@_i\tup{\dowa}j, \Gamma \vdash \Delta, @_i\tup{\dowa}\varphi}
    
    \tabularnewline[7pt]

    \prftree[r]{\footnotesize(${\tup{\cmpr}}$L)$^{\ddagger}$}
    {@_i\tup{\alpha}j, @_i\tup{\beta}k, \tagg{\cmpr}{j}{k}, \Gamma \vdash \Delta}
    {@_i\tup{\alpha \cmpr \beta}, \Gamma \vdash \Delta}

    \quad

    \prftree[r]{\footnotesize(${\tup{\cmpr}}$R)}
    {@_i\tup{\alpha}j, @_i\tup{\beta}k, \Gamma \vdash \Delta, @_i\tup{\alpha \cmpr \beta},\tagg{\cmpr}{j}{k}}
    {@_i\tup{\alpha}j, @_i\tup{\beta}k, \Gamma \vdash \Delta, @_i\tup{\alpha \cmpr \beta}}
    
    \tabularnewline[5pt]
    {\footnotesize$^{\dagger}$ $j$ is not in the conclusion}
    \qquad
    {\footnotesize $^{\ddagger}$ $j$ and $k$ are different and not in the conclusion}

  \end{tabular}

  \tabularnewline

  \midrule
  \textbf{\textsc{Rules for Data Comparison}} \\
  \midrule

  \tabularnewline[-7pt]
  \begin{tabular}{@{}c@{}}
    {\prftree[r]{\footnotesize(EqT)}
    {\tagg{=}{i}{i}, \Gamma \vdash \Delta}
    {\Gamma \vdash \Delta}}

    \quad

    {\prftree[r]{\footnotesize(Eq5)}
    {\tagg{=}{j}{k}, \tagg{=}{i}{j}, \tagg{=}{i}{k}, \Gamma \vdash \Delta}
    {\tagg{=}{i}{j}, \tagg{=}{i}{k}, \Gamma \vdash \Delta}}

	\quad

    {\prftree[r]{\footnotesize(NEqL)}
    {\Gamma \vdash \Delta, \tagg{=}{i}{j}}
    {\tagg{\neq}{i}{j}, \Gamma \vdash \Delta}}
    
    \quad

    {\prftree[r]{\footnotesize(NEqR)}
    {\tagg{=}{i}{j}, \Gamma \vdash \Delta}
    {\Gamma \vdash \Delta, \tagg{\neq}{i}{j}}}

  \end{tabular}
  
  \tabularnewline

  \midrule
  \textbf{\textsc{Structural Rules}} \\
  \midrule

  \tabularnewline[-7pt]
  \begin{tabular}{@{}c@{}}
    {\prftree[r]{\footnotesize(Cut)}
      {\Gamma \vdash \Delta, \varphi}
      {\varphi, \Gamma' \vdash \Delta'}
      {\Gamma, \Gamma' \vdash \Delta, \Delta'}}
    
    \quad

    {\prftree[r]{\footnotesize(WL)}
      {\phantom{\varphi,}\, \Gamma \vdash \Delta}
      {\varphi, \Gamma \vdash \Delta}}
    
    \quad

    {\prftree[r]{\footnotesize(WR)}
    {\Gamma \vdash \Delta\phantom{, \varphi}}
    {\Gamma \vdash \Delta, \varphi}}
  \end{tabular}
  
  \tabularnewline
  \bottomrule
  \end{tabular}
  }}
  \caption{Sequent Calculus $\gentzen$ for $\hxpd$.}\label[figure]{rules:hxpd}
\end{figure}

The inference rules of $\gentzen$ are organized into five groups, each corresponding to a different fragment of the language of $\hxpd$.
We briefly explain the rationale behind the rules in each group, highlighting the logical principles and the design choices they reflect.

\begin{description}
	\item[\textsc{Propositional Rules.}] The propositional rules govern implication and falsity, and follow familiar patterns from classical sequent calculi.
	Note that the node expression $\varphi$ in the axiom rule (Ax) is restricted in form. This restriction suffices for completeness. The generalized version of the rule where 
	$\varphi$ is of the form $@_i\psi$ for $\psi$ an arbitrary node expression can be obtained as a derived rule. \\[-1.5em]
	
	\item[\textsc{Rules for Nominals.}] The rules for nominals handle node expressions involving named points in the model, ensuring they behave as references to specific nodes. (@T) and (@5) characterize reflexivity and Euclideanness. The (Nom) rule allows a named node to be given a fresh alias. This plays a technical role in the completeness proof.
	Finally, the substitution rules (S$_m$) reflect that when $@_ij$ holds, then $i$ can be substituted by $j$ in different contexts. These rules are inspired by the treatment of equality in the labeled sequent calculus with equality from~\cite{Negri14}.
    Again, we note that all (S$_i$) are presented in a somewhat restricted form.
    Specifically, a more general version of (S$_1$) can be formulated without the restriction on $\varphi$.
    In turn, (S$_2$) can be generalized to handle arbitrary paths $\alpha$, not just atomic modalities $\dowa$.
    Finally, (S$_3$) can be generalized to handle arbitrary comparisons $\cmpr$, not just data equality.
    These generalized rules are derivable within the system.
    Presenting them in their restricted form helps to simplify the system by limiting the scope of interactions between rules, which is particularly beneficial for proving cut elimination.
    \\[-1.5em]
	
	\item[\textsc{Rules for Modalities.}] The modal rules capture the navigational fragment of the language.
	(@L) and (@R) reflect the global nature of the satisfiability operator: when the operator is nested, only the inner occurrence is relevant.
    The diamond rules (${\tup{\dowa}}$L) and (${\tup{\dowa}}$R) follow standard patterns from modal sequent calculi.
    In particular, (${\tup{\dowa}}$L) introduces fresh nominals as witnesses for the $\tup{\dowa}$ modality.
    Once again, one can consider generalized versions of the diamond rules in which the atomic modality $\dowa$ is replaced by an arbitrary path expression $\alpha$.
    These generalized forms are available as derived rules within the system.
    Finally, the rules $(\tup{\cmpr}\text{L})$ and $(\tup{\cmpr}\text{R})$ decompose a node expression $@_i\tup{\alpha \cmpr \beta}$ into its essential components: node expressions $@_i\tup{\alpha}j$ and $@_i\tup{\beta}k$, capturing the navigation paths, and the atomic comparison $\tup{j{:} \cmpr k{:}}$ at the endpoints of the paths.
    Notice there is a direct analogy between the behavior of these last two rules and the diamond rules.
    \\[-1.5em]

	\item[\textsc{Rules for Data Comparison.}] Data comparison rules handle equality and inequality between data values at the endpoints of paths.
	(EqT) and (Eq5) express that data equality is an equivalence relation, while (NeqL) and (NeqR) allow inequality to be reasoned about in terms of equality.  \\[-1.5em]

	\item[\textsc{Structural Rules.}] Finally, the structural rules govern manipulation of the sequent structure itself.
    These rules play a central role in the meta-theoretical analysis of the system.
	We will show in \Cref{sec:cut}, however, that these rules can be eliminated. 	
\end{description}

   \subsection{Soundness}
\label{sec:soundness}

We now turn to the proof of soundness for $\gentzen$.
We define a formal notion of sequent validity in the context of hybrid data models and then verify, via a case-by-case analysis, that each inference rule of $\gentzen$ preserves this notion of validity.
This strategy ensures that all provable sequents are valid.
We begin by formally defining the semantics of sequents in terms of the satisfaction relation introduced earlier.

\begin{definition}\label{def:seq:validity}
    A sequent $\Gamma \vdash \Delta$ is \emph{valid} iff for all $\amodel$, it follows that $\amodel \Vdash \Gamma$ implies $\amodel \Vdash \psi$ for some $\psi \in \Delta$.
    A rule preserves validity iff the validity of the premisses of the rule implies the validity of the conclusion of the rule.
\end{definition}


\begin{lemma}[Soundness]\label{lemma:soundness}
    Every rule in $\gentzen$ preserves validity.
\end{lemma}
\begin{proof}
    We present a selection of representative cases below. The remaining cases use a similar argument and can be verified by routine inspection.  In all cases below we reason by contradiction.
    \begin{description}
        \item[\textnormal{(Nom)}] 
            Let $\amodel[A]$ be a model s.t.\ $\amodel[A] \Vdash \Gamma$ and $\amodel[A] \nVdash \psi$ for all $\psi \in \Delta$.
            Then, introduce a new nominal $j$ that is not in $\Gamma,\Delta$ and build a model $\amodel[B]$ that is just like $\amodel[A]$ with the exception that $\amodel[B] \Vdash @_ij$.
            We have $\amodel[B] \Vdash @_ij, @_i\varphi, \Gamma$ and $\amodel[B] \nVdash \psi$ for all $\psi \in \Delta$.
            This contradicts the validity of the premiss of the rule.
        
        \item[\textnormal{($\tup{\cmpr}$L)}] We have: (1) $\cmpr$ is $=_{\compc}$, or (2) $\cmpr$ is $\neq_{\compc}$.
            For (1), take any model $\amodel[A]$ s.t.:
                $\amodel[A] \Vdash @_i\<\alpha =_{\compc} \beta\>, \Gamma$, and
                $\amodel[A] \nVdash \psi$ for all $\psi \in \Delta$.
            The semantics of $@_i\<\alpha =_{\compc} \beta\>$ tells us there are $n$ and $n'$ in $\amodel[A]$ s.t.:
                $\amodel[A], g(i), n \Vdash \alpha$,
                $\amodel[A], g(i), n' \Vdash \beta$, and
                $(n,n') \in {\approx_{\compc}}$.
            Then, choose nominals $j$ and $k$ that do not appear in $\Gamma,\Delta,\alpha,\beta$ and build a model $\amodel[B]$ that is identical to $\amodel[A]$ with the exception that $\amodel[B], n \Vdash j$ and $\amodel[B], n' \Vdash k$.
            It is clear that $\amodel[B] \Vdash @_i\<\alpha\>j, @_i\<\beta\>k, \tagg{=}{j}{k}, \Gamma$ and $\amodel[B] \nVdash \psi$ for all $\psi \in \Delta$.
            This contradicts the validity of the premiss of the rule.
            The case for  (2) is similar.
        
        \item[\textnormal{($\tup{\cmpr}$R)}] We have: (1) $\cmpr$ is $=_{\compc}$, or (2) $\cmpr$ is $\neq_{\compc}$.
            For (1), take any model $\amodel[A]$ s.t.:
                $\amodel[A] \Vdash @_i\<\alpha\>j, @_i\<\beta\>k, \Gamma$ and $\amodel[B] \nVdash \psi$ for all $\psi \in \Delta, @_i\<\alpha =_{\cmpr} \beta\>$.
            In particular, $\amodel[A] \nVdash @_i\<\alpha =_{\compc} \beta\>$.
            This means that for all $n$ and $n'$ in $\amodel[A]$, it follows that
                $\amodel[A], g(i), n \Vdash \alpha$,
                $\amodel[A], g(i), n' \Vdash \beta$, and
                $(n,n') \notin {\approx_{\compc}}$.
            This contradicts the validity of premiss of the rule.
            The case for (2) is similar.
            \qedhere
    \end{description}
\end{proof}

Soundness of $\gentzen$ follows from \Cref{lemma:soundness} by induction on the structure of a derivation of a sequent. 

\begin{theorem}[Soundness]\label{th:soundness}
    Every provable sequent in $\gentzen$ is valid, i.e., $\Gamma \vdash_\gentzen \Delta$ implies $\Gamma \vDash \Delta$.
\end{theorem}

   \subsection{Invertibility of Rules}
\label{sec:invertibility}

Having established the soundness of the calculus $\gentzen$, 
we prove one of its key properties: all inference rules are invertible.
Invertibility plays a central role in the proof-theoretic analysis of a sequent calculus. It allows backward reasoning---from a conclusion to its premisses---without loss of validity.
Invertible rules are also particularly well-suited for proof search procedures (see, e.g.,~\cite{Negri14}).
In our case, invertibility also plays a role in our proof of completeness of $\gentzen$, where we rely on the ability to apply certain rules in reverse to construct derivations..

\begin{definition}\label[definition]{def:invertible}
    A rule~$(\Rho)$ is \emph{invertible} iff there is a derivation of each premiss of the rule whose leaves are either axioms or the conclusion of the rule.
    Any such derivation is called an \emph{inverse} of the rule and is denoted by~$(\inv{\arule})$.
\end{definition}

\begin{theorem}\label{th:invertible}
    Every rule in $\gentzen$ is invertible.
\end{theorem}
\begin{proof}
    Invertibility of propositional rules is standard.
    For
    (@T), 
    (@5),
    (Nom),
    $(\text{S}_1)$,
    $(\text{S}_2)$,
    $(\text{S}_3)$,
    $(\tup{\dowa}\text{R})$,
    $(\tup{\cmpr}\text{R})$,
    (EqT), and
    (Eq5) 
    invertibility follows by weakening.
    We show ($\inv{@\text{L}}$), ($\inv{\tup{\dowa}\text{L}}$), and ($\inv{\tup{\cmpr}\text{L}}$).%
    \footnote{We use colors to guide derivations: \cc{cyan} for cut expressions, \rr{pink} for principal, and \pp{magenta} for other relevant elements.}

    \smallskip
    \noindent
        Case ($\inv{@\text{L}}$): Given a derivation of ${@_i\varphi, \Gamma \vdash \Delta}$, we build a derivation of $@_i\varphi$  as:
        $$\scalebox{\thescalefactor}
        {
            \prftree[r]{\footnotesize(Cut)}
            {
                \prftree[r]{\footnotesize($@$R)}
                {
                    \prfbyaxiom{\footnotesize(Ax)}
                    {@_i\varphi, \Gamma \vdash \Delta, @_i\varphi}
                }
                {@_i\varphi, \Gamma \vdash \Delta, \cc{@_j@_i\varphi}}
            }
            {
                \prfassumption
                {\cc{@_j@_i\varphi}, \Gamma \vdash \Delta}
            }
            {@_i\varphi, \Gamma \vdash \Delta}
        }$$
        Case ($\inv{\tup{\dowa}\text{L}}$):
        Given a derivation of ${@_i\<\dowa\>\varphi}, \Gamma \vdash \Delta$, we build a derivation of $\taag{\dowa}{i}{j}, @_j\varphi, \Gamma \vdash \Delta$ as:
        $$\scalebox{\thescalefactor}
        {
            \prftree[r]{\footnotesize(Cut)}
            {
                \prftree[r]{\footnotesize($\<\dowa\>$R)}
                {
                    \prfbyaxiom{\footnotesize(Ax)}
                    {\taag{\dowa}{i}{j}, @_j\varphi, \Gamma \vdash \Delta, @_i\<\dowa\>\varphi, @_j\varphi}
                }
                {\taag{\dowa}{i}{j}, @_j\varphi, \Gamma \vdash \Delta, \cc{@_i\<\dowa\>\varphi}}
            }
            {
                \prfassumption
                {\cc{@_i\<\dowa\>\varphi}, \Gamma \vdash \Delta}
            }
            {\taag{\dowa}{i}{j}, @_j\varphi, \Gamma \vdash \Delta}
        }$$
        Case (${\tup{\cmpr}\text{L}}$):
        Given a derivation of ${{@_i\tup{\alpha \cmpr \beta}}, \Gamma \vdash \Delta}$, we derive ${@_i\tup{\alpha}j, @_i\tup{\beta}k, \tagg{\cmpr}{j}{k}, \Gamma \vdash \Delta}$ as:
        $$\scalebox{\thescalefactor}
        {
            \prftree[r]{\footnotesize(Cut)}
            {
                \prftree[r]{\footnotesize($\<\cmpr\>$R)}
                {
                    \prfbyaxiom{\footnotesize(Ax)}
                    {@_i\tup{\alpha}j, @_i\tup{\beta}k, \tagg{\cmpr}{j}{k}, \Gamma \vdash \Delta, @_i\tup{\alpha \cmpr \beta}, \tagg{\cmpr}{j}{k}}
                }
                {@_i\tup{\alpha}j, @_i\tup{\beta}k, \tagg{\cmpr}{j}{k}, \Gamma \vdash \Delta, \cc{@_i\tup{\alpha \cmpr \beta}}}
            }
            {
                \prfassumption
                {\cc{@_i\tup{\alpha \cmpr \beta}}, \Gamma \vdash \Delta}
            }
            {@_i\tup{\alpha}j, @_i\tup{\beta}k, \tagg{\cmpr}{j}{k}, \Gamma \vdash \Delta}
        }\qedhere$$
\end{proof}

   \subsection{Completeness}
\label{sec:completeness}

In this section, we establish the completeness of~$\gentzen$.
Rather than constructing canonical models, we prove completeness by way of a Hilbert-style axiomatization for~$\hxpd$, which we refer to as~$\hilbert$.
The axiomatization, introduced in~\cite{ArecesF21}, has been shown to be both sound and complete.
Concretely, we demonstrate that every theorem of~$\hilbert$ corresponds to a provable sequent in~$\gentzen$.
This ensures that all semantically valid sequents are derivable in the sequent calculus.
Our strategy underscores the expressiveness of~$\gentzen$, and leverages existing results.
As a further benefit, we reveal a close correspondence between axiomatic and sequent-based reasoning in~$\hxpd$.

\begin{lemma}\label[lemma]{lemma:derivability}
    Let $\vdash_{\hilbert}^n \varphi$ indicate that $\varphi$ is a theorem in $\hilbert$, with a derivation of length $n$.
    In addition, let $i$ be a nominal not in $\varphi$.
    It follows that, for any $n$, $\vdash_{\hilbert}^n \varphi$ implies $\vdash_{\gentzen} @_i\varphi$ (i.e, $\vdash @_i\varphi$ is provable in $\gentzen$). 
\end{lemma}
\begin{proof}
    The proof proceeds by (strong) induction on the length of a derivation in $\hilbert$.
    The base case requires that each axiom $\varphi$ in $\hilbert$ has a corresponding provable sequent $\vdash_\gentzen @_i\varphi$.
    To simplify these proofs we will make use of the following derived rules in $\gentzen$. 
 
    \medskip
\noindent
\textbf{Derived Rules.} 
We introduce derived rules for treating other propositional connectives as if they were primitive.
Notably,
$$    \scalebox{\thescalefactor}
{
	
        {\prftree[r]{\small($\top$L)}
        {@_i\top, \Gamma \vdash \Delta}
        {\Gamma \vdash \Delta}}
        ~
        ~
        {\prftree[r]{\small($\land$L)}
        {@_i \varphi, @_i \psi, \Gamma \vdash \Delta }
        {@_i(\varphi \land \psi), \Gamma \vdash \Delta}}
        ~
        {\prftree[r]{\small($\land$R)}
        {\Gamma \vdash \Delta, @_i\psi}
        {\Gamma \vdash \Delta, @_i\varphi}
        {\Gamma \vdash \Delta, @_i(\varphi \land \psi)}}
        ~
        {\prftree[r]{\small($\liff$R)}
        {@_i\varphi, \Gamma \vdash \Delta, @_i\psi}
        {@_i\psi, \Gamma \vdash \Delta, @_i\varphi}
        {\Gamma \vdash \Delta, @_i(\varphi \liff \psi)}}
}
$$


\noindent
As mentioned earlier, we introduce a generalized form of (Ax) as a derived rule.
$$    \scalebox{\thescalefactor}
{
		\prfbyaxiom{\footnotesize(AxG)}{@_i\varphi, \Gamma \vdash \Delta, @_i\varphi}
%
}
$$
\noindent
In (AxG) $\varphi$ is an arbitrary node expression.
When no confusion arises, we will also refer to this rule simply as (Ax).
Finally, we introduce derived rules characterizing the symmetry of $\cmpr$ explicitly. Precisely,
 $$	    \scalebox{\thescalefactor}
 {
 	\prftree[r]{($\tup{\cmpr}$B).} 
        {\tup{j{:} \cmpr i{:}}, \Gamma \vdash \Delta}
        {\tup{i{:} \cmpr j{:}}, \Gamma \vdash \Delta}
    }$$

\noindent
\textbf{Base Case.} We are now ready to present derivations for selected axioms of $\hilbert$.
Beyond their technical role, these derivations also serve an explanatory purpose: they illustrate how the rules of $\gentzen$ reveal the semantic behavior of logical connectives, offering insight into their proof-theoretic interpretation and clarifying the structure of the corresponding axioms.
We focus on the axioms (equal), ($\cmpr$-comm), and ($\epsilon$-trans) expressing the equivalence properties of data equality: reflexivity, symmetry, and transitivity.
These axioms are listed below for reference.
\begin{align*}
    \text{(equal)}
        &~ \tup{\epsilon =_{\compc} \epsilon}
    \\
    \text{($\cmpr$-comm)}
        &~ \tup{\alpha \cmpr \beta} \liff \tup{\beta \cmpr \alpha}
    \\
    \text{($\epsilon$-trans)}
        &~ \tup{\alpha =_{\compc} \epsilon} \land \tup{\epsilon =_{\compc} \beta} \to \tup{\alpha =_{\compc} \beta}.
\end{align*}
For (equal), we must show
    $\vdash_{\gentzen} @_i\tup{\epsilon =_{\compc} \epsilon}$.
For ($\cmpr$-comm), we must show
    $\vdash_{\gentzen} @_i(\tup{\alpha \cmpr \beta} \liff \tup{\beta \cmpr \alpha})$.
Finally, for ($\epsilon$-trans), we must show
    $\vdash_{\gentzen} @_i(\tup{\alpha =_{\compc} \epsilon} \land \tup{\epsilon =_{\compc} \beta} \to \tup{\alpha =_{\compc} \beta})$.
We deal with each of these cases individually.


\medskip\noindent 
\textit{Reflexivity.} The derivation below establishes $\vdash_{\gentzen} @_i\tup{\epsilon =_{\compc} \epsilon}$.
        \begin{center}\scalebox{\thescalefactor}
            {
                \prftree[r]{\footnotesize(@T)}
                {
                    \prftree[r]{\footnotesize($\top$L)}
                    {
                        \prftree[r]{\footnotesize($\inv{{\land}\text{L}}$)}
                        {
                            \prftree[r]{\footnotesize($\tup{\cmpr}\text{R}$)}
                            {
                                \prftree[r]{\footnotesize(EqT)}
                                {
                                    \prfbyaxiom{\footnotesize(Ax)}
                                    {\tagg{=}{i}{i}, @_i\tup{\epsilon}i \vdash @_i\tup{\epsilon =_{\compc} \epsilon}, \tagg{=}{i}{i}}
                                }
                                {@_i\tup{\epsilon}i \vdash @_i\tup{\epsilon =_{\compc} \epsilon}, \tagg{=}{i}{i}}
                            }
                            {@_i\tup{\epsilon}i \vdash @_i\tup{\epsilon =_{\compc} \epsilon}}
                        }
                        {@_i\top, @_ii \vdash @_i\tup{\epsilon =_{\compc} \epsilon}}
                    }
                    {@_ii \vdash @_i\tup{\epsilon =_{\compc} \epsilon}}
                }
                {\vdash @_i\tup{\epsilon =_{\compc} \epsilon}}
            }
        \end{center}
        When reading the derivation bottom-up, most of the effort goes into constructing an empty path on the antecedent side of the sequent.
        Intuitively, the empty path encodes the idea that we remain at the current node.
        We then use ($\tup{\cmpr}$R) to compare the data at the node and itself.
        The derivation concludes with the rule (EqT) just before (Ax), which captures the reflexivity of data equality in $\gentzen$.

\medskip\noindent
    \textit{Symmetry.} The next derivation establishes $\vdash_{\gentzen} @_i(\tup{\alpha \cmpr \beta} \liff \tup{\beta \cmpr \alpha})$.
        \begin{center}\scalebox{\thescalefactor}
            {
                \prftree[r]{\footnotesize($\liff$R)}
                {
                    \prftree[r]{\footnotesize($\tup{\cmpr}$L)}
                    {
                        \prftree[r]{\footnotesize($\tup{\cmpr}$R)}
                        {
                            \prftree[r]{\footnotesize($\tup{\cmpr}$B)}
                            {
                                \prfbyaxiom{\footnotesize(Ax)}
                                {@_i\tup{\alpha}j, @_i\tup{\beta}k, \tagg{\cmpr}{k}{j} \vdash @_i\tup{\beta \cmpr \alpha},\tagg{\cmpr}{k}{j}}
                            }
                            {@_i\tup{\alpha}j, @_i\tup{\beta}k, \tagg{\cmpr}{j}{k} \vdash @_i\tup{\beta \cmpr \alpha},\tagg{\cmpr}{k}{j}}
                        }
                        {@_i\tup{\alpha}j, @_i\tup{\beta}k, \tagg{\cmpr}{j}{k} \vdash @_i\tup{\beta \cmpr \alpha}}
                    }
                    {@_i\tup{\alpha \cmpr \beta} \vdash @_i\tup{\beta \cmpr \alpha}}
                }
                {
                    \prftree[r]{\footnotesize($\tup{\cmpr}$L)}
                    {
                        \prftree[r]{\footnotesize($\tup{\cmpr}$R)}
                        {
                            \prftree[r]{\footnotesize($\tup{\cmpr}$B)}
                            {
                                \prfbyaxiom{\footnotesize(Ax)}
                                {@_i\tup{\beta}k,@_i\tup{\alpha}j, \tagg{\cmpr}{j}{k} \vdash @_i\tup{\alpha \cmpr \beta},\tagg{\cmpr}{j}{k}}
                            }
                            {@_i\tup{\beta}k,@_i\tup{\alpha}j, \tagg{\cmpr}{k}{j} \vdash @_i\tup{\alpha \cmpr \beta},\tagg{\cmpr}{j}{k}}
                        }
                        {@_i\tup{\beta}k,@_i\tup{\alpha}j, \tagg{\cmpr}{k}{j} \vdash @_i\tup{\alpha \cmpr \beta}}
                    }
                    {@_i\tup{\beta \cmpr \alpha} \vdash @_i\tup{\alpha \cmpr \beta}}
                }
                {\vdash @_i(\tup{\alpha \cmpr \beta} \liff \tup{\beta \cmpr \alpha})}
            }
        \end{center}
        As before, reading the derivation bottom-up, we observe that we first construct the appropriate paths in the antecedent.
        These paths serve to navigate to the nodes whose data values we aim to compare.
        As mentioned, in $\gentzen$ the derived rule ($\tup{\cmpr}$B) captures symmetry.
        The derivation concludes with an application of this rule before (Ax).

\medskip\noindent
    \textit{Transitivity.} The final derivation establishes $\vdash_{\gentzen} @_i(\tup{\alpha =_{\compc} \epsilon} \land \tup{\epsilon =_{\compc} \beta} \to \tup{\alpha =_{\compc} \beta})$.
        \begin{center}\scalebox{\thescalefactor}
            {
                \prftree[r]{\footnotesize($\to$R,$\land$L)}
                {
                    \prftree[r]{\footnotesize($\tup{\cmpr}$L)}
                    {
                        \prftree[r]{\footnotesize($\land$L,$\tup{\cmpr}$R)}
                        {
                            \prftree[r]{\footnotesize(@5,WL)}
                            {
                                \prftree[r]{\footnotesize($\text{S}_3$,WL)}
                                {
                                    \prftree[r]{\footnotesize($\tup{\cmpr}$B)}
                                    {
                                        \prftree[r]{\footnotesize(Eq5)}
                                        {
                                            \prfbyaxiom{\footnotesize(Ax)}
                                            {\tagg{=}{a}{b}, \tagg{=}{c}{a}, \tagg{=}{c}{b} \vdash \tagg{=}{a}{b}}
                                        }
                                        {\tagg{=}{c}{a}, \tagg{=}{c}{b} \vdash \tagg{=}{a}{b}}
                                    }
                                    {\tagg{=}{a}{c}, \tagg{=}{c}{b} \vdash \tagg{=}{a}{b}}
                                }
                                {@_{d}c, \tagg{=}{a}{c}, \tagg{=}{d}{b} \vdash \tagg{=}{a}{b}}
                                }
                            {@_ic, @_id, @_i\top, \tagg{=}{a}{c}, @_i\tup{\alpha}a, @_i\tup{\beta}b, \tagg{=}{d}{b} \vdash @_i\tup{\alpha =_{\compc} \beta},\tagg{=}{a}{b}}
                        }
                        {@_i\tup{\alpha}a, \rr{@_i\tup{\epsilon}c}, \tagg{=}{a}{c}, \rr{@_i\tup{\epsilon}d}, @_i\tup{\beta}b, \tagg{=}{d}{b} \vdash @_i\tup{\alpha =_{\compc} \beta}}
                    }
                    {@_i\tup{\alpha =_{\compc} \epsilon}, @_i\tup{\epsilon =_{\compc} \beta} \vdash @_i\tup{\alpha =_{\compc} \beta}}
                }
                {\vdash @_i(\tup{\alpha =_{\compc} \epsilon} \land \tup{\epsilon =_{\compc} \beta} \to \tup{\alpha =_{\compc} \beta})}
            }
        \end{center}
        The key idea is to chain two data equalities together.
        The derivation (going again bottom-up) begins by using empty paths $\epsilon$ to refer to the intermediate node that links $\alpha$ and $\beta$.
        Because of the side conditions on the rules being used, this intermediate node ends up having two names $c$ and $d$.
        The rule (@5) indicates $c$ and $d$ are alias.
        Given this information, the rule (S$_3$) replaces $d$ by $c$ in $\tagg{=}{d}{b}$, yielding $\tagg{=}{c}{b}$.
        The rule ($\tup{\cmpr}$B) prepares the sequent for the application of (Eq5), which captures transitivity in $\gentzen$, before ending with (Ax).
\medskip

To sum up, (equal), ($\cmpr$-comm), and ($\epsilon$-trans) are introduced in \cite{ArecesF21} to express the equivalence properties of data equality.
Arguably, rules (EqT) and (Eq5) make these properties more transparent.
The derivation in $\gentzen$ of all axioms in $\hilbert$ can be found in~\cite{arxiv}.

\medskip
\noindent\textbf{Inductive Case.}    
    We must show that $\vdash_{\hilbert}^m \psi$ implies $\vdash_{\gentzen} @_i \psi$. The inductive hypothesis (IH) is: for all $1 \leq n < m$, if $\vdash_{\hilbert}^{n} \varphi$, then $\vdash_{\gentzen} @_i\varphi$, for some nominal $i$ not occurring in~$\varphi$.
    We proceed by cases, depending on whether ${\vdash_{\hilbert}^m} \psi$ is obtained using one of the four inference rules in $\hilbert$: (MP), (Nec), (Name), or (Paste). The first three cases follow directly; details can be found in~\cite{arxiv}.
    Let us discuss the (Paste) rule:
    $$\scalebox{\thescalefactor}
    {
        \prftree[r]{{\footnotesize where $j$ and $k$ are different and not in $\chi$, $\alpha$, $\beta$.}}
        {\vdash @_j\<\dowa\>k \land \<k{:}\alpha \cmpr \beta\> \to \chi}
        {\vdash \<j{:}\dowa\alpha \cmpr \beta\> \to \chi}
    }
    $$
    
 \noindent    
    An application of this rule assumes that 
        for $n < m$ 
            $\vdash_{\hilbert}^{n} (@_j\<\dowa\>k \land \<k{:}\alpha \cmpr \beta\>) \to \chi$.
    From the IH, $\vdash_{\gentzen} @_i((@_j\<\dowa\>k \land \<k{:}\alpha \cmpr \beta\>) \to \chi)$.
    The next derivation establishes $\vdash_{\gentzen} @_i(\<j{:}\dowa\alpha \cmpr \beta\> \to \chi)$, as required.
    \vspace{-.5cm}
    \begin{center}\scalebox{\thescalefactor}
        {
        \prftree[r]{\footnotesize($\to$R)}
        {
            \prftree[r]{\footnotesize($\tup{\cmpr}$L)}
            {
            \prftree[r]{\footnotesize($@$L)}
            {
                \prftree[r]{\footnotesize($\tup{\dowa}$L)}
                {
                    \prftree[r]{\footnotesize(Cut)}
                    {
                        \prftree[r]{\footnotesize($\land$R,W$*$)}
                        {
                            \prftree[r]{\footnotesize($@$R)}
                            {
                            \prfbyaxiom{\footnotesize(Ax)}
                            {\taag{\dowa}{j}{k} \vdash @_j\tup{\dowa}k}
                            }
                            {\taag{\dowa}{j}{k} \vdash @_i@_j\tup{\dowa}k}
                        }
                        {
                            \prftree[r]{\footnotesize($\inv{@\text{L}}$)}
                            {
                                \prftree[r]{\footnotesize(${\tup{\cmpr}\text{R}}$,W$*$)}
                                {
                                    \prftree[r]{\footnotesize(Ax)}
                                    {\tagg{\cmpr}{a}{b} \vdash \tagg{\cmpr}{a}{b}}
                                }
                                {@_i\tup{k{:}\alpha}a, @_i\tup{\beta}b, \tagg{\cmpr}{a}{b} \vdash @_i\tup{k{:}\alpha \cmpr \beta}}
                            }
                            {\rr{@_k\tup{\alpha}a}, @_i\tup{\beta}b, \tagg{\cmpr}{a}{b} \vdash @_i\tup{k{:}\alpha \cmpr \beta}}
                        }
                        {\taag{\dowa}{j}{k},@_k\tup{\alpha}a, @_i\tup{\beta}b, \tagg{\cmpr}{a}{b} \vdash \cc{@_i(@_j\tup{\dowa}k \land \tup{k{:}\alpha \cmpr \beta})}}
                    }
                    {
                        \prftree[r]{\footnotesize($\inv{{\to}\text{R}}$)}
                        {
                            \prfassumption
                            {\pp{\vdash_{\gentzen} @_i(@_j\tup{\dowa}k \land \tup{k{:}\alpha \cmpr \beta} \to \chi)}}
                        }
                        {\cc{@_i(@_j\tup{\dowa}k \land \tup{k{:}\alpha \cmpr \beta})} \vdash @_i \chi}
                    }
                    {\taag{\dowa}{j}{k}, @_k\tup{\alpha}a, @_i\tup{\beta}b, \tagg{\cmpr}{a}{b} \vdash @_i \chi}
                }
                {@_j\tup{\dowa\alpha}a, @_i\tup{\beta}b, \tagg{\cmpr}{a}{b} \vdash @_i \chi}
            }
            {@_i\tup{j{:}\dowa\alpha}a, @_i\tup{\beta}b, \tagg{\cmpr}{a}{b} \vdash @_i \chi}
            }
            {@_i\tup{j{:}\dowa\alpha \cmpr \beta} \vdash @_i \chi}
        }
        {\pp{\vdash_{\gentzen} @_i(\tup{j{:}\dowa\alpha \cmpr \beta} \to \chi)}}
        } 
    \end{center}
    In the derivation above,  (W$*$) indicates the simultaneous application of (WL) and (WR).
\end{proof}

\begin{theorem}[Completeness]\label{th:completeness:gentzen}
	Every valid sequent is provable.
\end{theorem}
\begin{proof}
    Suppose $\gamma_1,\ldots,\gamma_n \vdash \delta_1,\ldots,\delta_m$ is valid.
    From the completeness result in~\cite{ArecesF21}, we know there is $k$ such that $\vdash^k_{\hilbert} \bigwedge_{1 \le i \le n} \gamma_i \to \bigvee_{1 \le j \le m} \delta_j$.
    From \Cref{lemma:derivability}, we get $\vdash_{\gentzen} @_i(\bigwedge_{1 \le i \le n} \gamma_i \to \bigvee_{1 \le j \le m} \delta_j)$.
    This implies $@_i\gamma_1, \ldots, @_i\gamma_n \vdash_{\gentzen} @_i\delta_1, \ldots,@_i\delta_m$, and so $\gamma_1,\ldots,\gamma_n \vdash_{\gentzen} \delta_1,\ldots,\delta_m$ as required.
\end{proof}

\section{Cut Elimination}\label{sec:cut}

The rule (Cut) plays a crucial role in our proof of the completeness of~$\gentzen$.
Specifically, it allows us to compose translations of Hilbert-style derivations within $\gentzen$.
Despite its utility, however, (Cut) can be eliminated.
The general strategy, following~\cite{Negri14}, is to push applications of (Cut) ``upwards'' in the derivation tree until they disappear at the level of leaves.
To this end, we perform an induction on the \emph{size} of the \emph{active cut expression}, with a sub-induction on the \emph{(Cut) height}.
We introduce these notions and related results below.

\begin{definition}\label{def:size}
    The $\size$ of a node expression is defined by mutual recursion as:
    \begin{align*}
        \size(p) & = 1 &
            \size(\varphi \to \psi) & = 1 + \size(\varphi) + \size(\psi) &
                \size(\dowa) & = 1 \\
        \size(i) & = 1 &
            \size(@_i\varphi) & = 1 + \size(\varphi)  &
                \size(i{:}) & = 1 \\
        \size(\bot) & = 1 &
            \size(\tup{\dowa} \varphi) & = 1 + \size(\varphi) &
                \size(\varphi?) & = 1 + \size(\varphi) \\
        &&
            \size(\tup{\alpha \cmpr \beta}) & = 1 + \size(\alpha) + \size(\beta) &
                \size(\alpha \beta) & = \size(\alpha) + \size(\beta)
    \end{align*}
    The function $\size$ induces a well-founded partial order over the set of node expressions.
\end{definition}

\begin{definition}
    The \emph{height} of a derivation is the length of its longest branch (e.g., a derivation consisting of only an application of (Ax) has height 1).
    If $\Gamma \vdash \Delta$ is the end-sequent in a derivation, we will use $\Gamma \vdash^n \Delta$ to indicate that the derivation has height $n$.
    The \emph{cut height} of an application of (Cut) in a given derivation is the sum of the heights of the derivations of the premisses of the rule; i.e., if we have derivations $\Gamma \vdash^n \Delta, \cc{\varphi}$, and $\cc{\varphi}, \Gamma' \vdash^m \Delta'$, using (Cut), we obtain a derivation $\Gamma, \Gamma' \vdash^{(\max(n,m)+1)} \Delta, \Delta'$. In this case, the cut height is $n+m$.
    In such an application of (Cut), we call $\cc{\varphi}$ the \emph{active cut} expression.
\end{definition}

Intuitively, the \emph{cut height} measures how close to the leaves of a derivation a particular application of (Cut) occurs, taking into consideration the derivations of \emph{both} premisses of the rule.
This will be important in our proof of cut elimination. 
%
%
%
We can now state and prove the main result of this section.

\begin{theorem}\label{th:cut:elimination}
    Every use of (Cut) in the derivation of a provable sequent can be eliminated.
\end{theorem}
\begin{proof}
    The proof is by induction on two measures: the size of the active cut expression, and the cut height.
    More precisely, in the proof we associate with each application of (Cut) in a derivation a pair $(k,h)$, called \emph{cut complexity}, where: $k$ corresponds to the size of the active cut expression, and $h$ corresponds to the cut height.
    The induction is on the lexicographic order of the pairs $(k,h)$.
    
    \medskip\noindent
    \textbf{Base Case.}
    The base cases of the induction are relatively direct.
    They involve derivations in which (Cut) is applied only once, with axiom rules applied to its premisses, i.e., the premisses of (Cut) must be instances of (Ax) or ($\bot$).
    Eliminating (Cut) in such configurations is unproblematic.
    We illustrate one representative case below.
    Suppose $\aderivation[C]$ and $\aderivation[A]$ are derivations with the following structure:
    $$\scalebox{\thescalefactor}
        {
            \prftree[r,l]{\footnotesize(Cut)}{$\aderivation[C]=$}
            {
                \prfbyaxiom{\footnotesize(Ax)}
                {\Gamma \vdash \Delta, \cc{\varphi}}
            }
            {
                \prfbyaxiom{\footnotesize($\bot$)}
                {\cc{\varphi}, \Gamma' \vdash \Delta'}
            }
            {\Gamma, \Gamma'  \vdash \Delta, \Delta'}

            \qquad
            \qquad

            \prftree[r,l]{\footnotesize($\Rho$)}{$\aderivation[A]=$}
            {}
            {\Gamma, \Gamma'  \vdash \Delta, \Delta'}
        }
    $$
    The derivation $\aderivation[C]$ represents a case in which one of the premisses of (Cut) is (Ax) and the other is ($\bot$).
    This derivation can be transformed into the cut-free derivation $\aderivation[A]$ in which: $(\Rho)$ is (Ax) if $\varphi \notin \Gamma$; and it is ($\bot$) otherwise.
    The remaining cases use a similar argument.

	\medskip\noindent
    \textbf{Inductive Step.}
    The key idea is to identify an application of (Cut) that is minimal according to the cut height, and eliminate it.
    More precisely, we proceed by considering in a derivation a sub-derivation ending in an application of (Cut) with complexity $(k,h)$ where the value for $h$ is minimal.
    This guarantees that no other instance of (Cut) appears above it in the sub-derivation.
    We show that this sub-derivation can be replaced by an alternative one which uses only (Cut) instances whose associated pairs $(k',h')$ are strictly smaller than $(k,h)$.
    We can then invoke a (strong) inductive hypothesis to claim that these (Cut) instances can also be eliminated. 
        
    We proceed by a case analysis based on the syntactic form of the active cut, and on whether the active cut is principal in either premiss of the application of (Cut).
    This analysis guides a systematic transformation of the derivation where: we push the (Cut) upwards, or replace it with applications of (Cut) involving smaller active cut expressions.
    In either case, the process eventually leads to a cut-free derivation.
    We illustrate how this transformation unfolds in some representative scenarios.

    \medskip
    \noindent
    {\it Non-principal Cases.}
    First, let us cover the case where the active cut is not principal in the right premiss of (Cut).
    In this case, the application of (Cut) is permuted up.
    To illustrate this process, consider, e.g., the derivation:
    $$\scalebox{\thescalefactor}
        {
            \prftree[r]{\footnotesize(Cut)}
            {
                \prfassumption
                {\Gamma \vdash^{n} {\Delta}, \cc{\varphi}}
            }
            {
                \prftree[r]{\footnotesize(${\tup{\cmpr}}$L)}
                {
                    \prfassumption
                    {@_i\tup{\alpha}j, @_i\tup{\beta}k, \tagg{\cmpr}{j}{k}, \varphi, {\Gamma'} \vdash^m \Delta}
                }
                {\cc{\varphi}, @_i\tup{\alpha \cmpr \beta}, {\Gamma'} \vdash \Delta}
            }
            {@_i\tup{\alpha \cmpr \beta}, \Gamma, {\Gamma'} \vdash {\Delta}, \Delta'}
        }
    $$

    \noindent
    Suppose that in this derivation, the use of (Cut) has an associated complexity  $(\size(\varphi), n + (m + 1))$, where $n+(m+1)$ is minimal.
    W.l.o.g., assume that $j$ and $k$ do not appear  in $\Gamma \vdash \Delta$.%
        \footnote{If $j$ or $k$ do appear in $\Gamma \vdash \Delta$, we can simply choose different nominals, and rewrite the derivation of ${{\varphi}, @_i\tup{\alpha \cmpr \beta}, {\Gamma'} \vdash \Delta}$ using the new selection of nominals.}
    We transform this derivation into:
    $$\scalebox{\thescalefactor}
        {
            \prftree[r]{\footnotesize(${\tup{\cmpr}}$L)}
            {
                \prftree[r]{\footnotesize(Cut)}
                {
                    \prfassumption
                    {\Gamma \vdash^{n} {\Delta}, \cc{\varphi}}
                    \quad
                }
                {
                    \prfassumption
                    {\cc{\varphi}, @_i\tup{\alpha}j, @_i\tup{\beta}k, \tagg{\cmpr}{j}{k}, {\Gamma'} \vdash^m \Delta}
                }
                {@_i\tup{\alpha}j, @_i\tup{\beta}k, \tagg{\cmpr}{j}{k}, \Gamma, {\Gamma'} \vdash {\Delta}, \Delta'}
            }
            {@_i\tup{\alpha \cmpr \beta}, \Gamma, {\Gamma'} \vdash \Delta, {\Delta'}}
        }
    $$

    \noindent
    The application of (Cut) in the transformed derivation has complexity $(\size(\varphi), n+m)$, so it can be eliminated by the inductive hypothesis.

    The remaining cases where the active cut is not principal in the right premiss follow the same strategy.
    To see why, note that all such derivations have the following general structure:
    $$\scalebox{\thescalefactor}
        {
            \prftree[r]{\footnotesize(Cut)}
            {
                \prfassumption
                {\Gamma \vdash^{n} \Delta,\cc{\varphi}}
            }
            {
                \prftree[r]{\footnotesize($\Rho$)}
                {
                    \prfassumption
                    {{\varphi},\Phi', \Gamma' \vdash^{m} \Delta',\Sigma'}
                }
                {\cc{\varphi},\Phi, \Gamma' \vdash \Delta',\Sigma}
            }
            {\Phi, \Gamma, \Gamma'  \vdash \Delta, \Delta', \Sigma}
        }
    $$
    In this derivation, we are considering ($\Rho$) is a single premiss rule of $\gentzen$,
    the set $\varphi,\Gamma',\Delta'$ is the context for the rule, and
    the set $\Phi',\Phi,\Sigma',\Sigma$ are the node expressions the rule acts upon.
    Again, we assume that (Cut) has complexity $(\size(\varphi),n+(m+1))$ where $n+(m+1)$ is minimal; i.e., where (Cut) is not used in ${\Gamma \vdash^{n} \Delta, \cc{\varphi}}$, nor in ${{\varphi},\Phi', \Gamma' \vdash^{m} \Delta',\Sigma'}$.
    Modulo a possible renaming of nominals, we transform this derivation into:
    $$\scalebox{\thescalefactor}
        {
            \prftree[r]{\footnotesize($\Rho$)}
            {
                \prftree[r]{\footnotesize(Cut)}
                {
                    \prfassumption
                    {\Gamma \vdash^{n} \Delta, \cc{\varphi}}
                }
                {
                    \prfassumption
                    {\cc{\varphi},\Phi',\Gamma' \vdash^{m} \Delta',\Sigma'}
                }
                {\Phi', \Gamma', \Gamma \vdash \Delta, \Delta', \Sigma'}
            }
            {\Phi, \Gamma', \Gamma \vdash \Delta, \Delta', \Sigma}
        }
    $$
    The use of (Cut) in the transformed derivation has complexity $(\size(\varphi), n+m)$, using the inductive hypothesis, we obtain a cut-free derivation of ${\Phi, \Gamma', \Gamma \vdash \Delta, \Delta', \Sigma}$.

    The cases where the active cut is not principal in the left premiss is symmetric. The ($\to$L) case---the only two-premise rule in $\gentzen$---is handled similarly, and is well known in the literature.
    
    \medskip
    \noindent
    {\it Principal Cases.}
    Let us now turn our attention to derivations where the active cut is principal in both premisses.
    Such cases are central to the proof and require careful handling to ensure that the application of (Cut) can still be pushed upwards, and ultimately eliminated.
    We must examine all combinations of rules with a possibly matching active cut.
    We illustrate a few representative cases below.

    Let us consider first the interaction between ($\tup{\dowa}$R) and ($\tup{\dowa}$L).
    We adapt the strategy presented in~\cite{Negri14}.
    Suppose that in a derivation we encounter an application of (Cut) of minimal height of the form:
    $$\scalebox{\thescalefactor}
        {
            \prftree[r]{\footnotesize(Cut)}
            {
                \prftree[r]{\footnotesize($\tup{\dowa}$R)}
                {
                    \prfassumption
                    {@_i\tup{\dowa}j, \Gamma \vdash^{n} \Delta, @_i\tup{\dowa}\varphi, @_j\varphi}
                }
                {@_i\tup{\dowa}j,\Gamma \vdash \Delta, \cc{@_i\tup{\dowa}\varphi}}
            }
            {
                \prftree[r]{\footnotesize($\tup{\dowa}$L)}
                {
                    \prfassumption
                    {@_i\tup{\dowa}j, @_j\varphi, \Gamma' \vdash^{m} \Delta'}
                }
                {\cc{@_i\tup{\dowa}\varphi}, \Gamma' \vdash \Delta'}
            }
            {@_i\tup{\dowa}j,\Gamma, \Gamma' \vdash \Delta,\Delta'}
        }
    $$

    \noindent
    The exhibited (Cut) has complexity $(\size(@_i\tup{\dowa}\varphi), (n+1)+(m+1))$.
    We proceed to transform the derivation into a new one that reduces the complexity and moves us closer to a cut-free derivation.
    $$\scalebox{\thescalefactor}
        {
            \prftree[r]{\footnotesize(Cut$_2$)}
            {
                \prftree[r]{\footnotesize(Cut$_1$)}
                {
                    \prfassumption
                    {@_i\tup{\dowa}j, \Gamma,\vdash^{n} \Delta,@_j\varphi, \cc{@_i\tup{\dowa}\varphi}}
                }
                {
                    \prftree[r]{\footnotesize($\tup{\dowa}$L)}
                    {
                        \prfassumption
                        {@_i\tup{\dowa}j, @_j\varphi, \Gamma' \vdash^{m} \Delta'}
                    }
                    {\cc{@_i\tup{\dowa}\varphi}, \Gamma' \vdash \Delta'}
                }
                {@_i\tup{\dowa}j, \Gamma,\Gamma' \vdash \Delta,\Delta', \cc{@_j\varphi}}
            }
            {
                \prfassumption
                {\cc{@_j\varphi},@_i\tup{\dowa}j, \Gamma' \vdash^{m} \Delta'}
            }
            {@_i\tup{\dowa}j,\Gamma, \Gamma' \vdash \Delta,\Delta'}
        }
    $$ 
    \noindent
    In the transformed derivation, there are two applications of (Cut), labeled (Cut$_1$) and (Cut$_2$).  We reason as follows, (Cut$_1$), is the only (Cut) in its subderivation and has complexity $(\size(@_i\tup{\dowa}\varphi), n+(m+1))$.
    Therefore, by the inductive hypothesis, (Cut$_1$) can be eliminated resulting in a cut-free derivation of ${@_i\tup{\dowa}j, \Gamma, \Gamma' \vdash^h \Delta, \Delta', \cc{@_j\varphi}}$, for some unknown $h$.
    We use this cut-free derivation as a building block to construct the derivation:
    $$
    \scalebox{\thescalefactor}
        {
            \prftree[r]{\footnotesize(Cut$_2$)}
            {
                \prfassumption
                {@_i\tup{\dowa}j, \Gamma,\Gamma' \vdash^h \Delta,\Delta', \cc{@_j\varphi}}
                \quad
            }
            {
                \prfassumption
                {\cc{@_j\varphi},@_i\tup{\dowa}j, \Gamma' \vdash^{m} \Delta'}
            }
            {@_i\tup{\dowa}j,\Gamma, \Gamma' \vdash \Delta,\Delta'}
        }
    $$
    
    \noindent
    We can now see that (Cut$_2$) can also be eliminated.
    It has complexity $(\size(@_j\varphi), h+m)$ and it is the only (Cut) application in the derivation.
    Since $\size(@_j\varphi) < \size(@_i\tup{\dowa}\varphi)$, we can apply the inductive hypothesis, and obtain a fully cut-free derivation of ${@_i\tup{\dowa}j,\Gamma, \Gamma' \vdash \Delta,\Delta'}$.
    
    Let us now consider a case involving data comparisons: the interaction between ($\tup{\cmpr}$R) and ($\tup{\cmpr}$L).
    Namely, suppose that in a derivation we encounter an application of (Cut) of minimal height of the form:
    $$\scalebox{\thescalefactor}
        {                   
            \prftree[r]{\footnotesize(Cut)}
            {
                \prftree[r]{\footnotesize(${\tup{\cmpr}}$R)}
                {
                    \prfassumption
                    {@_i\tup{\alpha}j, @_i\tup{\beta}k, \Gamma \vdash^{n} \Delta, @_i\tup{\alpha \cmpr \beta},\tagg{\cmpr}{j}{k}}
                }
                {@_i\tup{\alpha}j, @_i\tup{\beta}k, \Gamma \vdash \Delta, \cc{@_i\tup{\alpha \cmpr \beta}}}
            }
            {
                \prftree[r]{\footnotesize(${\tup{\cmpr}}$L)}
                {@_i\tup{\alpha}j, @_i\tup{\beta}k, \tagg{\cmpr}{j}{k}, \Gamma' \vdash^m \Delta'}
                {\cc{@_i\tup{\alpha \cmpr \beta}}, \Gamma' \vdash \Delta'}
            }
            {@_i\tup{\alpha}j, @_i\tup{\beta}k, \Gamma, \Gamma' \vdash \Delta, \Delta'}
        }
    $$
    
    \noindent
    We transform this derivation into:
    $$\scalebox{\thescalefactor}
        {
            \prftree[r]{\footnotesize(Cut$_2$)}
            {
                \prftree[r]{\footnotesize(Cut$_1$)}
                {
                    \prfassumption
                    {@_i\tup{\alpha}j, @_i\tup{\beta}k, \Gamma \vdash^{n} \Delta, \tagg{\cmpr}{j}{k}, \cc{@_i\tup{\alpha \cmpr \beta}}}
                }
                {
                    \hspace*{-.3cm}
                    \prftree[r]{\footnotesize(${\tup{\cmpr}}$L)}
                    {@_i\tup{\alpha}j, @_i\tup{\beta}k, \tagg{\cmpr}{j}{k}, \Gamma' \vdash^m \Delta'}
                    {\cc{@_i\tup{\alpha \cmpr \beta}}, \Gamma' \vdash \Delta'}
                }
                {@_i\tup{\alpha}j, @_i\tup{\beta}k, \Gamma, \Gamma' \vdash \Delta, \Delta', \cc{\tagg{\cmpr}{j}{k}}}
            }
            {
                \hspace*{-1.3cm}
                \prfassumption
                {\cc{\tagg{\cmpr}{j}{k}}, @_i\tup{\alpha}j, @_i\tup{\beta}j, \Gamma' \vdash^m \Delta'}
            }
            {@_i\tup{\alpha}j, @_i\tup{\beta}k, \Gamma, \Gamma' \vdash \Delta, \Delta'}
        }
    $$
    
    \noindent
    The argument is similar, the original use of (Cut) has complexity $(\size(@_i\tup{\alpha \cmpr \beta}), (n+1)+(m+1))$.
    When we transform the derivation, we introduce two new applications of (Cut), labeled (Cut$_1$) and (Cut$_2$).
    (Cut$_1$) has complexity $(\size(@_i\tup{\alpha \cmpr \beta}), n+(m+1))$ and can be eliminated.  Then, (Cut$_2$), can be eliminated from the resulting derivation as it involves a simpler active cut expression.

    For the particular to $\gentzen$ case, let us consider the interaction between the rules ($\tup{\dowa}$R) and ($\tup{\cmpr}$R).
    Namely, suppose that in a derivation we encounter an application of (Cut) of minimal height of the form:
    $$\scalebox{\thescalefactor}
        {
            \prftree[r]{\footnotesize(Cut)}
            {
                \prftree[r]{\footnotesize(${\tup{\dowa}}$R)}
                {
                    \prfassumption
                    {@_i\tup{\dowa}j, \Gamma \vdash^{n} \Delta, @_i\tup{\dowa}a, @_ja}
                }
                {@_i\tup{\dowa}j, \Gamma \vdash \Delta, \cc{@_i\tup{\dowa}a}}
                \hspace{-2pt}
            }
            {
                \prftree[r]{\footnotesize(${\tup{\cmpr}}$R)}
                {
                    \prfassumption
                    {@_i\tup{\dowa}a, @_i\tup{\beta}b, \Gamma' \vdash^m \Delta', @_i\tup{\dowa \cmpr \beta},\tagg{\cmpr}{a}{b}}
                }
                {\cc{@_i\tup{\dowa}a}, @_i\tup{\beta}b, \Gamma' \vdash \Delta', @_i\tup{\dowa \cmpr \beta}}
            }
            {@_i\tup{\dowa}j, @_i\tup{\beta}b, \Gamma, \Gamma' \vdash \Delta, \Delta', @_i\tup{\dowa \cmpr \beta}}
        }
    $$

    \noindent
    Unlike the other cut-elimination cases we have examined, this particular case involves two right rules applied in the premisses of (Cut), rather than a mix of left and right rules.
    This synchronization of right rules is unusual, but it does not pose a problem.
    We can still transform the derivation to eliminate (Cut).

    \begin{flushleft}
        ~~
        \scalebox{\thescalefactor}
        {
                $\aderivation = {@_i\tup{\dowa}j, \Gamma \vdash^{n} \Delta, @_ja, \cc{@_i\tup{\dowa}a}}$
        }
    \end{flushleft}
    \vspace{-1cm}
    \begin{center}\scalebox{\thescalefactor}
        {
            \prftree[r]{\footnotesize(Cut$_2$)}
            {
                \prftree[r]{\footnotesize(Cut$_1$)}
                {
                    ~~
                    \prfassumption
                    {\aderivation}
                    \qquad
                }
                {
                    \prftree[r]{\footnotesize(${\tup{\cmpr}}$R)}
                    {
                        \prfassumption
                        {@_i\tup{\dowa}a, @_i\tup{\beta}b, \Gamma' \vdash^m \Delta', @_i\tup{\dowa \cmpr \beta},\tagg{\cmpr}{a}{b}}
                    }
                    {\cc{@_i\tup{\dowa}a}, @_i\tup{\beta}b, \Gamma' \vdash \Delta', @_i\tup{\dowa \cmpr \beta}}
                }
                {@_i\tup{\dowa}j, @_i\tup{\beta}b, \Gamma, \Gamma' \vdash \Delta, \Delta', @_i\tup{\dowa \cmpr \beta}, \cc{@_ja}}
            }
            {
                \prftree[r]{\footnotesize($\text{S}_2$)}
                {
                    \prftree[r]{\footnotesize(WL)}
                    {
                        \prftree[r]{\footnotesize(WL)}
                        {
                            \prftree[r]{\footnotesize(${\tup{\cmpr}}$R)}
                            {
                                \prfassumption
                                {@_i\tup{\dowa}a, @_i\tup{\beta}b, \Gamma' \vdash^m \Delta', @_i\tup{\dowa \cmpr \beta},\tagg{\cmpr}{a}{b}}
                            }
                            {{@_i\tup{\dowa}a}, @_i\tup{\beta}b, \Gamma' \vdash \Delta', @_i\tup{\dowa \cmpr \beta}}
                        }
                        {@_i\tup{\dowa}j, @_i\tup{\dowa}a, @_i\tup{\beta}b, \Gamma' \vdash \Delta', @_i\tup{\dowa \cmpr \beta}}
                    }
                    {@_ja, @_i\tup{\dowa}j, @_i\tup{\dowa}a, @_i\tup{\beta}b, \Gamma' \vdash \Delta', @_i\tup{\dowa \cmpr \beta}}
                }
                {\cc{@_ja}, @_i\tup{\dowa}j, @_i\tup{\beta}b, \Gamma' \vdash \Delta', @_i\tup{\dowa \cmpr \beta}}
            }
            {@_i\tup{\dowa}j, @_i\tup{\beta}b, \Gamma, \Gamma' \vdash \Delta, \Delta', @_i\tup{\dowa \cmpr \beta}}
        }
    \end{center}

    \noindent
    Just as in earlier cases, the original (Cut) is replaced by two new applications: (Cut$_1$), which has strictly smaller cut height, and (Cut$_2$), which involves a simpler active cut.
    Both of these are less complex than the original in the lexicographic sense, so we can apply the inductive hypothesis to eliminate them.
    This shows that even when both rules in the (Cut) are right rules, the same general strategy still works.
    The remaining cases of interaction of rules in cut-elimination can be found in \cite{arxiv}.
\end{proof}

\section{A Sequent System for Hybrid Logic}
\label{sec:properties} \label{sec:torben}


It is helpful to view $\hxpd$ as a modular extension of the Basic Hybrid Logic $\hylo(@)$ with additional constructs for data comparison.
This naturally raises the question: can we recover a sequent calculus $\rgentzen$ for $\hylo(@)$ by simply removing the rules related to data comparisons from $\gentzen$? In this section, we make this correspondence precise by showing that the resulting fragment of $\gentzen$ faithfully simulates the sequent calculus for $\hylo(@)$ presented in~\cite{Torben2011}.


\begin{definition}
    A \emph{formula of $\hylo(@)$} is a node expression without data comparisons, and 
    where diamonds are only of the form $\tup{\dowa}\varphi$.
    A \emph{hybrid model for $\hylo(@)$} is obtained by dropping the $\{\approx_{\compc}\}_{\compc \in \cmp}$ component from the hybrid data models in \Cref{def:models}.
    The semantics of formulas on hybrid models is as in \Cref{def:semantics}.
%
    Let $\rgentzen$ be obtained from $\gentzen$ by restricting sequents to formulas in $\hylo(@)$ of the form $@_i\varphi$, and dropping all rules involving data comparisons.
\end{definition}




The calculus $\rgentzen$ simulates the sequent calculus for $\hylo(@)$ in~\cite{Torben2011}, which we refer to as $\tgentzen$


\begin{lemma}\label{lemma:completeness:hylo}
    Every provable sequent in $\tgentzen$ is also provable in $\rgentzen$.
\end{lemma}

\begin{proof}
    We establish this result by showing that every rule in $\tgentzen$ is a derived rule~in~$\rgentzen$.
    The axioms in $\tgentzen$ align exactly with those in $\rgentzen$, establishing a one-to-one correspondence. Additionally, ($\to$L), ($\to$R), ($@$L), and ($@$R) are present in both calculi.
	The rule (Ref) from $\tgentzen$ is also present in $\rgentzen$ under the name ($@$T).
    The rules ($\land$L) and ($\land$R) are included as basic rules in $\tgentzen$, in contrast, they are derived in $\rgentzen$.
	Similarly, the rules ($[\dowa]$L) and ($[\dowa]$R), and (Nom$_1$) and (Nom$_2$), listed below, which are primitive in $\tgentzen$, are also derivable in $\rgentzen$.
	This shows that every derivation in $\tgentzen$ can be simulated in $\rgentzen$, completing the correspondence between the two systems.
	Details can be found in~\cite{arxiv}.

	\begin{center}\scalebox{\thescalefactor}{
		\begin{tabular}{@{}c@{}}
			\prftree[r]{\small$(\mathrm{Nom}_1)$}
			{\Gamma \vdash \Delta, @_i j}
			{\Gamma \vdash \Delta, @_i \varphi}
			{\Gamma \vdash \Delta, @_j \varphi}
			
			\quad
			
			\prftree[r]{\small$(\mathrm{Nom}_2)$}
			{\Gamma \vdash \Delta, @_ij}
			{\Gamma \vdash \Delta, \taag{\dowa}{i}{k}}
			{\taag{\dowa}{j}{k}, \Gamma \vdash \Delta}
			{\Gamma \vdash \Delta}
			
			\tabularnewline[10pt]
			
			\prftree[r]{\small$({[\dowa]}\mathrm{L}_1)$}
			{\Gamma \vdash \Delta, @_i\<\dowa\>j}{@_j\varphi, \Gamma \vdash \Delta}
			{@_i [\dowa]\varphi, \Gamma \vdash \Delta}
			
			\quad
			
			\prftree[r]{\small$({[\dowa]}\mathrm{R})$ $j$ is new.}
			{@_i\<\dowa\>j, \Gamma \vdash \Delta,@_j\varphi}
			{\Gamma \vdash \Delta, @_i [\dowa]\varphi}
		\end{tabular}
	}
\end{center}

As a concrete illustration of the correspondence between $\tgentzen$ and $\gentzen$, we show how the rule Nom$_2$ from $\tgentzen$ can be obtained as a derived rule within $\gentzen$.

$$\scalebox{\thescalefactor}
	{
		\prftree[r]{\footnotesize(Cut)}
		{\Gamma \vdash \Delta, \cc{@_ij}}
		{
			\prftree[r]{\footnotesize(Cut)}
			{
				\prftree[r]{\footnotesize(WL)}
				{\Gamma \vdash \Delta,@_i\tup{\dowa}k}
				{@_ij,\Gamma \vdash \Delta, \cc{@_i\tup{\dowa}k}}
			}
			{
				\prftree[r]{\footnotesize(S$_1$)}
				{
					\prftree[r]{\footnotesize(WL)}
					{@_j\tup{\dowa}k,\Gamma \vdash \Delta}
					{@_j\tup{\dowa}k, @_ij, @_i\tup{\dowa}k, \Gamma \vdash \Delta}
				}
				{\cc{@_i\tup{\dowa}k},@_ij, \Gamma \vdash \Delta}
			}
			{\cc{@_ij},\Gamma \vdash \Delta}
		}
		{\Gamma \vdash \Delta}
	}
$$

\noindent
We can see from the derivation how (Nom$_1$) is directly related to (S$_1$) in $\gentzen$.
\end{proof}


\begin{theorem}[Soundness and Completeness]\label{th:completeness:hylo}
    A sequent in $\hylo(@)$ is valid in all hybrid models iff it is provable in~$\rgentzen$.
\end{theorem}

In conclusion, the system $\tgentzen$ in~\cite{Torben2011} is derived from a natural deduction system for $\hylo(@)$ presented in the same work.
To reflect the structure and constraints of such a natural deduction system, $\tgentzen$ does not include weakening rules and, more importantly, is designed to be cut-free.
This comes with a trade-off: the rules (Nom$_1$), (Nom$_2$), and ($[\dowa]$L) in $\tgentzen$ are not invertible in $\tgentzen$.
In contrast, as a subsystem of $\gentzen$, $\rgentzen$ is fully invertible and enjoys cut elimination.
We bring attention also to the fact that, while (Cut) is admissible in $\tgentzen$, its use in a derivation cannot be systematically removed.
Thus, by supporting both admissibility and effective elimination of (Cut), $\rgentzen$ offers a more robust and analytically tractable framework for $\hylo(@)$.

\section{Final Comments}
\label{sec:final}

In this article, we introduce the Gentzen-style sequent calculus $\gentzen$ for the logic  Hybrid XPath with Data ($\hxpd$). $\hxpd$ is a hybrid modal logic
that captures not only the navigational core of XPath but also data comparisons, node labels (keys),
and key-based navigation operators.
In $\gentzen$, we provide rules for handling complex path expressions (composition and tests), as well as mechanisms for data comparisons (equality and inequality), nominals, and satisfiability operators.

We establish the completeness of $\gentzen$ by leveraging the completeness of the Hilbert-style calculus $\hilbert$ for $\hxpd$ introduced in~\cite{ArecesF21}.
We also show that every rule in $\gentzen$ is invertible, a crucial property in proof search. The system $\gentzen$ includes weakening rules (WL) and (WR), and the (Cut) rule.  (WL) and (WR) are never needed in derivations (they are only used to simplify sequents and eliminate irrelevant formulas at a given point in a proof).  Moreover, we showed that
(Cut) is also irrelevant for completeness, and we provide a cut elimination algorithm. 
Finally, we discuss how to define a subcalculus of $\gentzen$, which is a sound and complete sequent calculus for the Basic Hybrid Logic $\hylo(@)$, inheriting rule invertibility and cut elimination. This system improves on previously known sequent calculi for $\hylo(@)$ (see~\cite{Torben2011}).

Similarly to the approach in~\cite{Torben2011}, we can extend $\gentzen$ by incorporating pure axioms and existential saturation rules. A general result can be proved, showing that these extended calculi are sound and complete with respect to a large family of frame classes (see \cite{arxiv} for further details). 
There is a cost though: rule invertibility and cut-elimination are not guaranteed in these extensions.
Future work will explore whether these properties can be preserved under certain reasonable conditions.

Many aspects of XPath have been extensively investigated from a logical perspective, but their proof-theoretical properties have remained largely unexplored.
To our knowledge, the only exception is the sequent calculus for DataGL from~\cite{Baelde2016}.
In brief, DataGL can be seen as a formalization of a highly restricted fragment of XPath in which data comparisons are allowed only between the evaluation point and its strict descendants in a data tree.  In comparison, $\gentzen$ provides much more expressivity. 

The proof theory of modal logics has always been challenging.
One reason for this is that modal languages combine a simple propositional syntax with a rich relational semantics.
As a result, proof-theoretical approaches that rely solely on the \emph{form} of logical expressions struggle to capture their full expressive power.
Various proof-theoretical techniques—--such as display calculi, hypersequents, deep inference, labeled systems, and hybridization—--have been explored to establish a well-rounded proof theory for modal logics.
Data-aware modal logics extend the expressive power of traditional modal languages even further, incorporating mechanisms to handle pairs of complex paths in graphs and data comparisons.
Our results indicate that hybridization techniques (i.e., internalizing the use of labels via nominals and satisfiability operators) can provide the basis for a solid proof theory for modal logics with data-comparison operators.


There are several open lines for future research. 
In particular, we would like to prove the termination of our calculus and define an optimal proof search strategy.
Moreover, we wish to extend the calculus to more expressive fragments---such as those involving sibling relations, transitive closure navigation, or novel forms of data comparison, and to study proof theoretical aspects of an intuitionistic variant of $\hxpd$ (see~\cite{ACDF23}).

\bibliographystyle{eptcs}
\bibliography{biblio}




\end{document}